\documentclass[11pt, draftcls, onecolumn]{IEEEtran}
\usepackage{amsmath, etoolbox}
\usepackage{amsthm}
\usepackage{amsfonts}
\usepackage{amssymb}
\usepackage{amscd}
\usepackage{cite}
\usepackage[dvips]{graphicx}
\def\BibTeX{{\rm B\kern-.05em{\sc i\kern-.025em b}\kern-.08em
    T\kern-.1667em\lower.7ex\hbox{E}\kern-.125emX}}
\newtheorem{theorem}{Theorem}
\usepackage{subfigure}
\usepackage{newlfont}
\usepackage{enumerate}
\usepackage{amssymb}
\usepackage{latexsym}
\usepackage{epsfig}
\usepackage{color}
\usepackage{mathtools}
\usepackage{algorithmic}
\usepackage{multicol}
\newtheorem{lemma}{Lemma}
\newtheorem{proposition}{Proposition}
\newtheorem{corollary}{Corollary}
\newcounter{hints}
\renewcommand{\thehints}{\alph{hints}}
\newcommand{\hintedrel}[2][]{%
  \stepcounter{hints}%
  \if\relax\detokenize{#1}\relax\else\csxdef{hint@#1}{\thehints}\fi
  \mathrel{\overset{\textrm{(\thehints)}}{\vphantom{\le}{#2}}}%
}
\newcommand{\restarthintedrel}{\setcounter{hints}{0}}
\newcommand{\hintref}[1]{\csuse{hint@#1}}

\newcommand{\lb}{\left(}
\newcommand{\rb}{\right)}
\newcommand{\ls}{\left[}
\newcommand{\rs}{\right]}

\newcommand{\ux}{\underline{\mathbf{x}}}
\newcommand{\uy}{\underline{\mathbf{y}}}
\newcommand{\uz}{\underline{\mathbf{z}}}
\newcommand{\matX}{\mathbf{X}}
\newcommand{\mc}[1]{\mathcal{#1}}

\long\def\symbolfootnote[#1]#2{\begingroup%
\def\thefootnote{\fnsymbol{footnote}}\footnote[#1]{#2}\endgroup}

\makeatletter
\renewcommand\subsubsection{\@startsection{subsubsection}{3}{0mm}{0ex plus 0.1ex minus 0.1ex}%
{0.3ex plus 0ex}{\normalfont\normalsize\itshape}}%
\makeatother

\title{On Finding a Subset of Non-Defective Items from a Large Population\thanks{This work was presented in part in \cite{Sharma_ITA_2013}.}}

\author{ Abhay Sharma and Chandra R. Murthy\\
  Dept. of ECE,  Indian Institute of Science, Bangalore 560 012, India \\
  abhay.bits@gmail.com, cmurthy@ece.iisc.ernet.in
         }

\begin{document}
\maketitle
\begin{abstract}

In this paper, we derive mutual information based upper and lower bounds on the number of nonadaptive group 
tests required to identify a given number of ``non-defective'' items from a large population 
containing a small number of ``defective'' items. 
We show that a reduction in the number of tests is achievable 
compared to the approach of first identifying all the defective items and then picking the required number 
of non-defective items from the complement set. 
In the asymptotic regime with the population 
size $N \rightarrow \infty$, to identify $L$ non-defective items out of a population containing $K$ defective 
items, when the tests are reliable, our results show that $\frac{C_s K}{1-o(1)} (\Phi(\alpha_0, \beta_0) + o(1))$ 
measurements are sufficient, where $C_s$ is a constant independent of $N, K$ and $L$, and $\Phi(\alpha_0, \beta_0)$ is a bounded function of $\alpha_0 \triangleq \lim_{N\rightarrow \infty} \frac{L}{N-K}$ and $\beta_0 \triangleq \lim_{N\rightarrow \infty} \frac{K}{N-K}$. 
Further, in the nonadaptive group testing setup, we obtain rigorous upper and lower bounds 
on the number of tests under both dilution and additive noise models. Our results are derived using a general sparse signal model, by virtue of which, they are also 
applicable to other important sparse signal based applications such as compressive sensing.
\end{abstract}


\section{Introduction}
\label{sec:Introduction}

Sparse signal models are of great interest due to their applicability in a variety of  
areas such as compressive sensing\cite{Candes_DecbyLP}, 
group testing\cite{Dorfman_GT, du2006}, 
signal de-noising\cite{Bruckstein_sparse}, 
subset selection\cite{Trop_JustRelax}, etc. Generally speaking, in a sparse
signal model, out of a given number $N$ of input variables, 
only a small subset of size $K$  contributes to 
the observed output.
For example, in a non-adaptive group testing setup, the output
depends only on whether the items from the defective set participate
or not participate in the group test.
Similarly, in a compressive sensing setup, the output signal 
is a set of random projections of the signal corresponding to the 
non-zero entries (support set) of the input vector.
This \emph{salient} subset of inputs is referred to by different 
names,
e.g., defective items, sick individuals, support set, etc.
In the sequel, we will refer to it as \emph{the active set}, and its
complement as \emph{the inactive set}. 
In this paper, we address the issue of the \emph{inactive subset recovery}.
That is, we focus on the task of finding an $L~(\le N-K)$ sized subset of \emph{the inactive set} (of size $N-K$),
given the observations from a sparse signal model with $N$ inputs,
out of which $K$ are active.

The problem of finding a subset of items belonging to the inactive set is  of 
interest in many applications.
An example is the spectrum hole search problem
in the cognitive radio (CR) networks\cite{HaykinCognitiveRadio}.
It is well known that the primary user occupancy (active set) is
sparse in the frequency domain over a wide band of 
interest\cite{Cabric05cognitive,FCCreport}. 
To setup a CR network, the secondary users need to find 
an appropriately wide unoccupied (inactive) frequency band.
Thus, the main interest here is the identification of \emph{only
a sub-band} out of the total available unoccupied band, i.e., it 
is an inactive subset recovery problem.
Furthermore, the required bandwidth of the spectrum hole
will typically be a small fraction of the entire bandwidth that is free at any 
point of time\cite{Sharma_Globecom_2012}.
Another example is a product manufacturing plant, where
a small shipment of non-defective (inactive) items has to be delivered
on high priority. Once again, the interest here is
on the identification of a subset of the non-defective items using as few tests as possible.

Related work: In the group testing literature, the problem of bounding the number of tests required to identify the defective items in a large pool has been studied, both in the 
noiseless and noisy settings, both for tractable decoding algorithms as well as under general information theoretic models
\cite{kautz_nonrand_supimp, Erdos_cover, Ruszinko94, Dachkov_bound, Sebo1985, Gilbert_GT, vetterli_ngt, malyutov_1, malyutov_2, malyutov_3, Atia_BooleanCS, jaggi_gtalgo, aldrige_cap_adaptive_gt, aldridge_almostseparable, Scar_Cev16}.
A combinatorial approach has been adopted in \cite{kautz_nonrand_supimp, Erdos_cover,Ruszinko94}, where 
explicit constructions for the test matrices are used, e.g., using superimposed codes, 
to design matrices with properties that lead to guaranteed detection of a small number of defective
items.
Two such properties were  considered: disjunctness and separability\cite{du2006}.\footnote{A test matrix, with tests indexing the rows and items indexing 
the columns, is said to be $k$-disjunct 
if the boolean sum of every $k$ columns does not equal any other column in the matrix.
Also, a test matrix is said to be $k$-separable if the boolean sum of every set of $k$ columns is unique.}
A probabilistic approach was adopted in\cite{Dachkov_bound,Sebo1985, Gilbert_GT, vetterli_ngt}, where  random test matrix designs were considered, and 
upper and lower bounds on the number of tests required to satisfy the properties 
of disjunctness or separability with high probability were derived. In particular, \cite{vetterli_ngt} analyzed the performance of group testing under the so-called dilution noise. 
Another study \cite{jaggi_gtalgo} uses random test designs, and develops computationally efficient algorithms for identifying 
defective items from the noisy test outcomes by exploiting the connection with compressive sensing.
A very recent work \cite{Scar_Cev16} uses novel information theoretic techniques, based on information density, to study 
the phase transitions for Bernoulli test matrix designs and measurement-optimal recovery algorithms.
A general sparse signal model for studying group testing problems, that turns out  to be 
very useful in dealing with noisy settings, was proposed and used 
in \cite{malyutov_1, malyutov_2, malyutov_3, Atia_BooleanCS}.
In this framework, the group testing problem was formulated as a detection problem and
a one-to-one correspondence was established with a communication channel model.
Using information theoretic arguments,  mutual information based expressions (that are easily computable for a wide variety of noisy channels) for upper and lower bounds on 
the number of tests were obtained\cite{Atia_BooleanCS}.
In the related field of compressive sensing, an active line of research has 
focused on the conditions under which reliable signal recovery from observations 
drawn from a linear sparse signal model is possible, 
for example, conditions on the number of measurements required 
and on isometry properties of the measurement matrix (\!\!\cite{Candes_Stable,WW_SharpThresholds}, 
and references therein). In particular, there exists a good understanding of the 
bounds on the number of measurements required for support recovery from noisy 
linear projections (e.g., \cite{saligrama_CS, fletcher2009necessary, 
Reeves_conf_supprecov,Jin_SuppRecovMAC, Wainwright_SparseRecovery}). 


Thus, to the best of our knowledge, fundamental bounds on the number of tests needed to find 
$L$ non-defective items, which is the focus of this paper, have not been derived in 
the existing literature. 
A recent work \cite{Yoo_arxiv_2013} studies the problem of finding zeros in a sparse vector 
in the framework of compressive sensing. The authors propose computationally efficient 
recovery algorithms and study their performance through simulations. In contrast, our 
work builds on our earlier work \cite{Sharma_ITA_2013}, and focuses on deriving information 
theoretic upper and lower bounds on the number of measurements needed for identifying a 
given number of inactive items in a large population with arbitrarily small probability of error. 
In this paper, we consider the general sparse signal model employed
in  \cite{malyutov_1, Atia_BooleanCS} in context of a support recovery problem.
The model consists of $N$ input covariates, out of which, an unknown subset $S$
of size $K$ are ``active''; in the sense that, only the active variables,
i.e., the variables from the set $S$, are relevant to the output. 
Mathematically, this is modeled by assuming that, given the active set $S$, the 
output $Y$ is independent of remaining input variables. Further, the probability 
distribution of the output conditioned on a given active set,  is assumed to be 
known for all possible active sets.
Given multiple observations from the this model, we propose and analyze decoding schemes
to identify \emph{a} set of $L$ inactive variables.
We compare two alternative decoding schemes: (a) Identify
the active set and then choose $L$ inactive covariates
randomly from the complement set, and, (b) Decode the inactive subset 
directly from the observations.
Our main contributions are as follows:
\begin{enumerate}
  \item We analyze the average probability of error for both the decoding schemes. 
    We use the analysis to obtain mutual information based upper bounds on the 
    number of observations required to identify a set of $L$ inactive 
    variables with the probability of error decreasing exponentially with the number of observations.
  \item We specialize the above bounds to various
    noisy non-adaptive group testing scenarios, and characterize the 
    number of tests required to identify $L$ non-defective items, in terms of $L$, $N$ and $K$.
  \item We also derive a lower bound, based on Fano's inequality,
    characterizing the number of observations required to identify $L$ inactive 
    variables.
\end{enumerate}

Our results show that, compared to the conventional approach of
identifying the inactive subset by first identifying the active set, 
directly searching for an $L$-sized inactive subset offers a  reduction in the 
number of observations (tests/measurements), especially when $L$ is small compared to $N-K$.
When the tests are reliable, in the asymptotic regime as $N \rightarrow \infty$, if $\frac{L}{N-K} \rightarrow \alpha_0$ and $\frac{K}{N-K} \rightarrow \beta_0$, $\frac{C_s K}{1-o(1)} (\Phi(\alpha_0, \beta_0) + o(1))$ 
measurements are sufficient, where $C_s$ is a constant independent of $N, K$ and $L$, and $\Phi(\alpha_0, \beta_0)$ is a bounded function of $\alpha_0$ and $\beta_0$. 
%
We show that this improves on the number of observations required by the conventional approach, in the sequel. 

The rest of the paper is organized as follows.
Section \ref{sec:PrbSetup} describes the signal model
and problem setup. 
We present our upper and lower bounds on 
the number of observations in Sections~\ref{result_suff} and \ref{necc_cond}, 
respectively.
An application of the bounds to group testing is described in
Section~\ref{sec_group_test}.  
The proofs for the main results are provided in Section~\ref{sec_thm_proofs}, 
and concluding remarks are offered in Section~\ref{sec_conclusions}.

\noindent \textbf{Notation:} 
For any positive integer $a$, $[a] \triangleq \{1, 2, \ldots, a\}$. 
For any set $A$, $A^c$ denotes complement operation and $|A|$ denotes 
the cardinality of the set. For any two sets $A$ and $B$, 
$A\backslash B = A \cap B^c$, i.e., elements of $A$ that are not in $B$. 
$\{ \emptyset\}$ denotes the null set.
Scalar random variables (RVs) are represented by capital 
non-bold alphabets, e.g., $\{Z_1, Z_3, Z_5, Z_8\}$ 
represent a set of $4$ scalar RVs. 
If the index set is known, 
we also use the index 
set as sub-script, e.g., $Z_{S}$, where $S = \{1, 3, 5, 8\}$. Bold-face 
letters represent random matrices (or a set of vector random variables).
We use an index set to specify a subset of columns from the given 
random matrix.
For example, let $\mathbf{Z}$ denote a random matrix with $n$ columns. For
any $S \subset [n]$, $\mathbf{Z}_{S}$ denotes a set of $|S|$ columns of 
$\mathbf{Z}$. Similarly, for any $S_1, S_2 \subset [n]$, $\mathbf{Z}_{S_1 \cup S_2}$
denotes a set of columns of $\mathbf{Z}$ specified by the index set $S_1 \cup S_2$.
Individual vector RVs are also denoted 
using an underline, e.g., $\underline{\mathbf{z}}$ represents a single random vector.
For any discrete random variable $Z$, $\{Z\}$ represents the set of all realizations
of $Z$. Similarly, for a random matrix $\mathbf{Z}$, whose entries are discrete random variables, 
$\{ \mathbf{Z} \}$ represents the set of all realizations of $\mathbf{Z}$.
For any two jointly distributed random variables $\uz_1$ and $\uz_2$, with a slight
abuse of notation, let $P(\uz_1|\uz_2)$ denote the conditional probability 
distribution of $\uz_1$ given ``a realization $\uz_2$'' of the random variable $\uz_2$.
Similarly $P(\uz_1|\mathbf{Z})$ denote the conditional probability 
distribution of $\uz_1$, given a realization $\mathbf{Z}$ of the random matrix $\mathbf{Z}$.
$\mathcal{B}(q), q \in [0~1]$ denotes the Bernoulli distribution with parameter
$q$.
$\mathbb{I}_{\mathcal{A}}$ denotes the indicator function, which returns $1$ if 
the event $\mathcal{A}$ is true, and returns $0$ otherwise.
Note that, $x(n)=O(y(n))$ implies that $\exists~ B > 0$ and $n_0 > 0$, such that $|x(n)| \le B |y(n)|$ for all $n > n_0$.
Similarly, 
$x(n)=\Omega(y(n))$ implies that $\exists~ B > 0$ and $n_0 > 0$, such that $|x(n)| \ge B |y(n)|$ for all $n > n_0$.
Also, $x(n) = o(y(n))$ implies that for every $\epsilon > 0$, there exists an $n_0 > 0$ such that $|x(n)| \le \epsilon |y(n)|$ for all $n > n_0$.
In this work, unless otherwise specified, all logarithms  to the base $e$.
For any $p \in [0, 1]$, $H_b(p)$ denotes the binary entropy in nats, i.e.,
$H_b(p) \triangleq -p \log(p) - (1-p) \log(1-p)$.

\section{Problem Setup}
\label{sec:PrbSetup}
In this section, we describe the signal model and problem setup.
Let $X_{[N]} = \bigl [ X_1, X_2, \ldots, X_N \bigr ]$ denote a set of $N$ independent
and identically distributed input random variables (or \emph{items}).
Let each $X_j$ belong to a finite alphabet denoted by $\mathcal{X}$ and be
distributed as $\text{Pr}\{X_j = x\} = Q(x), x \in \mathcal{X}$,
$j=1, 2, \ldots, N$. For a group of input variables, e.g., $X_{[N]}$,
$Q(X_{[N]}) = \prod_{j \in [N]} Q(X_j)$ denotes the known joint distribution
for all the input variables.
We consider a sparse signal model where only a subset of the input 
variables are
\emph{active} (or \emph{defective}), in the sense that only 
a subset of the input variables contribute to the output.
Let $S \subset [N]$ denote the set of input variables
that are active, with $|S| = K$. 
We assume that $K$, i.e., the size of the active set, is known.
Let $S^c \triangleq [N]\backslash S$ denote the
set of variables that are \emph{inactive} (or \emph{non-defective}).  
Let the output belong to a finite alphabet denoted by $\mathcal{Y}$. 
We assume that $Y$ is generated according to a known conditional distribution $P(Y|X_{[N]})$.
Then, in our observation model, we assume that given 
the active set, $S$, the output signal, $Y$, is independent 
of the other input variables.
That is, $~\forall~ Y \in \mathcal{Y}$,
\begin{align} \label{eq:active_def}
P(Y|X_{[N]}) = P(Y|X_S).
\end{align}

We observe the outputs corresponding to $M$ independent realizations 
of the input variables, and denote the inputs and the 
corresponding observations by $\{\matX, \uy\}$. 
Here, $\matX$ is an $M \times N$ matrix, with its $i^{\text{th}}$ row 
representing the $i^{\text{th}}$ realization of the input variables, and 
$\uy$ is an $M\times 1$ vector, with its $i^{\text{th}}$ component representing
the $i^{\text{th}}$ observed output.  Note that, the independence assumption across 
the input variables and across different observations implies that each entry 
in $\matX$ is independent and identically distributed (i.i.d.). 
Let $L \le N - K$. We consider the problem of finding \emph{a set} of $L$ inactive variables 
given the observation set, $\{\matX, \uy\}$. That is, we wish to find an index set 
$S_{H} \subset S^c$ such that $|S_{H}| = L$.
In particular, our goal is to derive information theoretic 
bounds on the number of observations (measurements/group tests) 
required to find a set of $L$ inactive variables
with the probability of error exponentially decreasing with
the number of observations. 
Here,  an error event occurs 
if the chosen inactive set contains one or more active variables.
Now, one way to find $L$ inactive variables is to find all the
active variables and then choose any $L$ variables from the complement set. Thus, existing bounds on $M$ for finding the active set are an upper
bound on the number of observations required for solving our problem.
However, intuitively speaking, fewer observations should 
suffice to find $L$ inactive variables, since we do not need to find the full active set.
This is confirmed by our results presented in the next section.

The above signal model can be equivalently described, see
Figure~\ref{figure:chan_mod}, using Shannon's random codebook based channel coding framework.
The active set $S$, that corresponds to one of the ${N \choose K}$ 
possible active sets with $K$ variables, constitutes the 
input message. Let $\matX \in  \mathcal{X}^{M \times N}$ be a random
codebook consisting of $N$ codewords of length $M$; each
associated with one of the $N$ input variables. Let $\ux_i$ denote the
codeword associated with $i^{\text{th}}$ input variable.
The encoder encodes the message as a length-$M$ message 
$\matX_S \in \mathcal{X}^{M \times K}$, that comprises of $K$ codewords, each 
of length $M$, chosen according to the index set $S$ from $\matX$. That
is, $\matX_S = [\ux_{i_1} ~ \ux_{i_2} \ldots \ux_{i_K}]$, for each $i_l \in S$.
Let $\matX_S^{(i)}$ denote the $i^{\text{th}}$ row of the matrix $\matX_S$ and let
$\uy(i)$ denote its $i^{\text{th}}$ component.
The encoded message is transmitted through a discrete memoryless 
channel\cite{Gallager_bk,Cover_bk}, denoted by $(\mathcal{X}^M, P(\uy|\matX_S), \mathcal{Y}^M)$, 
where $P(\uy|\matX_S) = \prod_{i=1}^M P(\uy(i)|\matX_S^{(i)})$ and  the distribution function 
$P(\uy(i)|\matX_S^{(i)})$ is known for each active set $S$.
Given the codebook $\matX$ and the output message $\uy$,  
our goal is to find \emph{a set} of $L$ variables \emph{not} belonging to the 
active set $S$.
We would like to mention briefly that the above signal model, 
proposed and used earlier in~\cite{Atia_BooleanCS, malyutov_1}, is a generalization of the signal models 
employed in some of the popular non-adaptive measurement system signal models such as 
compressed sensing\footnote{Although, in this work, we focus on models with finite 
alphabets, our results easily extend to models with continuous 
alphabets\cite{Atia_MutInf, saligrama_sparsesignproc}.} and non-adaptive group testing. 
Thus, the general mutual information based bounds on number of observations to 
find a set of inactive items obtained using the above model are applicable in a 
variety of practical scenarios.

We now discuss the above signal model in context of a specific
non-adaptive measurement system, namely the random pooling based, noisy 
non-adaptive group testing framework\cite{Atia_BooleanCS, du2006}.
In a group testing framework\cite{Atia_BooleanCS, du2006, malyutov_1}, we 
have a population of $N$ items, out of which $K$ are defective. 
Let $\mathcal{G} \subset [N]$ denote the
defective set, such that $| \mathcal{G} | = K$. The group tests are defined
by a boolean matrix, $\matX \in \{0,1\}^{M \times N}$, that assigns
different items to the $M$ group tests (pools). 
In the $i^{\text{th}}$ test, the items corresponding to the columns 
with $1$ in the $i^{\text{th}}$ row of $\matX$ are tested. 
Thus, $M$ tests are specified. As in \cite{Atia_BooleanCS}, we consider an i.i.d.\ random Bernoulli
measurement matrix, where each $X_{ij} \sim \mathcal{B}(p)$ for some
$0 < p < 1$. Here, $p$ is a design parameter 
that controls the average group size.
If the tests are completely reliable, then the output of the $M$ tests is given by the 
boolean OR of the columns of $\matX$ corresponding to the \emph{defective set} $\mathcal{G}$.
We consider the following two different noise models\cite{vetterli_ngt,Atia_BooleanCS}: 
(a) An \emph{additive} noise model, where there is a 
probability, $q \in (0,1] $, that the outcome of a group test containing
only non-defective items comes out positive;  (b) A \emph{dilution} model, where there 
is a probability, $u \in (0,1]$, that a given item does not participate
in a given group test.
We would like to mention that although we consider the two popular noise models mentioned above, the results of this paper can be adapted to other noise models also.
Let $\underline{d}_i \in \{0,1\}^M$. Let 
$\underline{d}_i(j) \sim \mathcal{B}(1-u)$ be chosen independently 
for all $j=1, 2, \ldots M$ and for all $i=1, 2, \ldots N$. Let 
$\mathbf{D}_i \triangleq \text{diag}(\underline{d}_i)$.
Let ``$\bigvee$'' denote the boolean OR operation.
The output vector $\underline{y} \in \{0, 1\}^M$ can be represented as
\begin{align} \label{eq:gtmodel}
  \underline{y} = \bigvee_{i \in \mathcal{G}} \mathbf{D}_i 
  \underline{x}_i \bigvee \underline{w},
\end{align}
where $\underline{x}_i \in \{0, 1\}^M$ is the $i^{\text{th}}$ column of $\matX$, 
$\underline{w} \in \{0, 1\}^M$ is the additive noise with the $i^{\text{th}}$ component
$\underline{w}(i) \sim \mathcal{B}(q)$.
For the noiseless case, $u=0, q=0$. In an additive model, $u=0, q > 0$.
In a dilution model, $u>0, q =0$.

The above ``logical-OR'' signal model represents many practical 
non-adaptive group testing measurement systems.
For example, consider a medical screening application, where a large number
of individuals need to be tested for the presence of a specific antigen in 
their blood. The blood samples drawn from the different individuals 
are pooled together, according to a randomly generated test matrix $\matX$ (as described above),  
into multiple pools. Each pool is tested for the presence of the specific
antigen. This test is well modeled by the OR-operation described above, i.e., 
when the tests are reliable, a test outcome is positive if one or more
samples in the pool contain the antigen, and, a test outcome is negative only if
none of the samples in the pool contain the antigen. Note that, given the knowledge
of the set of individuals with the presence of the antigen, 
the test outcome does not depend upon whether the blood sample from any other
individual is included in the pool or not.
For several other examples of the above described measurement system,
see \cite{Gilber_DataStream, vetterli_ngt, Macula_data_pattern, Wolf_multi_access,du2006}.

We now relate this model with the general sparse signal model described above. 
Note that, $\mathcal{X} = \{0, 1\}$, $\mathcal{Y} =  \{0, 1\}$. 
Each item in the group testing framework corresponds to one of the $N$ input
covariates. The $i^{\text{th}}$ \emph{row} of the test matrix, which specifies
the $i^{\text{th}}$ random pool, corresponds to the $i^{\text{th}}$ realization
of the input covariates. From (\ref{eq:gtmodel}), 
given the defective set $\mathcal{G}$, the $i^{\text{th}}$ test outcome $\uy(i)$ 
is independent of values of input variables from the set $[N]\backslash \mathcal{G}$.
That is, with regards to test outcome,  it is \emph{irrelevant} whether the items 
from the set $[N]\backslash \mathcal{G}$ are included in the test or not.
Thus, $\mathcal{G}$ corresponds to the active set $S$.
Further, with regards to the channel coding setup, the test matrix $\matX$ corresponds 
to the random codebook, and each column specifies the $M$ length random code with the
associated item.  The channel model, i.e., the probability distribution 
functions $P(\uy | \matX_{\mathcal{G}})$ for any $\mathcal{G}$, is fully determined 
from (\ref{eq:gtmodel}) and the statistical models for the dilution and additive noise.
Thus, it is easy to see that the group testing framework is a special case of 
the general sparse model that we have considered, and, the number of group tests 
correspond directly to the number of observations in the context of sparse models. 



  
We now define two quantities that are very useful in the development to follow. Let $S$ be a given active set. For any $1 \le j \le K$, let $S^{(j)}$ and $S^{(K-j)}$ represent a partition of $S$ such that 
    $S^{(j)} \cup S^{(K-j)} = S$, $S^{(j)} \cap S^{(K-j)} = \{ \emptyset\}$ and $|S^{(j)}| = j$. 
    Define
    \begin{align} \label{eq:E0_def}
      E_0(\rho,j, n) &= - \log \sum_{Y \in \mathcal{Y}} ~ ~ \sum_{X_{S^{(K-j)}} \in \mathcal{X}^{K-j}} 
    \left \{ \sum_{X_{S^{(j)}} \in \mathcal{X}^j} Q(X_{S^{(j)}}) 
    \left( P(Y, X_{S^{(K-j)}} | X_{S^{(j)}})\right)^{\frac{1}{1 + \rho n}}
    \right \}^{ 1 + \rho n} 
 \end{align}
 for any positive integer $n$ and any $\rho \in [0,1]$.
    Define $I^{(j)} \triangleq I(Y, X_{S^{(K-j)}} ; X_{S^{(j)}} ) =  I(Y ; X_{S^{(j)}} | X_{S^{(K-j)}})$
    as the mutual information between $\{Y, X_{S^{(K-j)}}\}$ 
    and $X_{S^{(j)}}$\cite{Gallager_bk,Cover_bk}. Mathematically,
    \begin{align} \label{eq:MI_gen_def}
      I^{(j)} = \sum_{Y \in \mathcal{Y}}
      \sum_{X_{S^{(K-j)}} \in \mathcal{X}^{K-j}} 
      \sum_{X_{S^{(j)}} \in \mathcal{X}^j} 
      P(Y, X_{S^{(K-j)}} | X_{S^{(j)}}) Q(X_{S^{(j)}})
      \log  \frac{P(Y, X_{S^{(K-j)}} | X_{S^{(j)}})}{P(Y, X_{S^{(K-j)}}) }.
    \end{align}
    Using the independence assumptions in the signal model, by the symmetry of the codebook construction, 
    for a given $j$, $E_0(\rho, j, n)$ and $I^{(j)}$ are independent of the 
    specific choice of $S$, and of the specific partitions of $S$. 
    It is easy to verify that 
    $\dfrac{d E_0(\rho, j, n)}{d\rho} |_{\rho = 0} = n I^{(j)}$.
    Furthermore, it can be shown that $E_0(\rho, j,n)$ is a concave
    function of $\rho$\cite{Gallager_bk}.
    

\begin{figure}[t]
\centering
\includegraphics[scale=0.7]{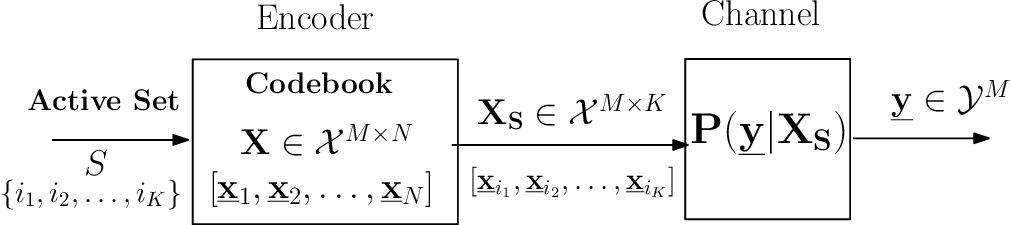}
\caption{Sparse signal model: An equivalent random codebook based channel coding model.}
\label{figure:chan_mod}
\end{figure}

\section{Sufficient Number of Observations} \label{result_suff}
We first present results on the sufficient number of observations to find a set of 
$L$ inactive variables.
The general methodology used to find the upper bounds
is as follows: (a) Given a set of inputs and 
observations, $\{\matX, \uy\}$, we first propose a decoding algorithm to find 
an $L$-sized inactive set, $S_{H}$; (b) For the given decoding scheme,
we find (or upper bound) the average probability of error, where the 
error probability is averaged over the random set $\{\matX, \uy\}$ as well as 
over all possible choices for the active set.
An error event occurs when the decoded set of $L$ inactive variables contains
one or more active variables. That is, with $S$ as the active set and
$S_H$ as the \emph{decoded} inactive set, an error occurs if 
$S \cap S_H \neq \{\emptyset\}$; 
(c) We find the relationships between $M$, $N$, $L$ and $K$ that will drive the average probability of error to zero.
Section \ref{DecSch1} describes the straightforward decoding 
scheme where we find the inactive variables by first isolating the 
active set followed by choosing the inactive set randomly from the
complement set. This is followed by the analysis of a new decoding scheme we propose
in Section \ref{DecSch2}, where we directly search for an inactive subset
of the required cardinality.


\subsection{Decoding scheme 1: Look into the Complement Set} \label{DecSch1}
One way to find a set of inactive (or non-defective) variables
is to first decode the active (defective) set and then pick a set of $L$ variables
uniformly at random from the complement set. 
Here, we employ maximum likelihood based optimal decoding \cite{Atia_BooleanCS} to find the active set. 
Intuitively, even if we choose a wrong active set, 
there is still a nonzero probability of picking a correct inactive set, since there remain only a few
active variables in the complement set.
We refer to this decoding scheme as the ``indirect'' decoding scheme.
The probability of error in identifying the active set was 
analyzed in \cite{Atia_BooleanCS}. The error probability when the same 
decoding scheme is employed to identify a inactive subset is similar, 
with an extra term to account for the probability of picking an incorrect
set of $L$ variables from the complement set. 
For this
decoding scheme, we present the following result, without proof, as a corollary 
to (Lemma III.I, \cite{Atia_BooleanCS}).

  \begin{corollary} \label{corr1_indir}
Let $N$, $M$, $L$ and $K$ be as defined above. 
For any $\rho \in [0~1]$, with the above decoding scheme, the average 
probability of error, $P_e$, in finding $L$ inactive variables is upper bounded as
\begin{align} \label{eq:nonasym_pe_cond0}
  P_e  & \le \max_{1 \le j \le K} \exp \left \{ 
    - \left ( M E_0(\rho,j,1) - { \rho \log \left[{N-K \choose j}C_0(j)\right] } - 
    \log \left[K {K \choose j} \right]  \right ) \right \},
\end{align}
where
$ C_0(j) \triangleq 1 - \mathop{\prod}_{l=0}^{L-1} \lb 1 - \frac{j}{N-L-l} \rb$
denotes the probability of error in choosing a set of $L$ inactive variables 
uniformly at random from a set of $N-K$ variables containing $j$ active variables.
\end{corollary}
From above, by lower bounding $E_0(\rho, j, 1)$ for any specific signal model,
we can  obtain a bound that gives us the sufficient number of observations to find a set 
of $L$ inactive variables. We obtain the corresponding bound in the context of
non-adaptive group testing in Section~\ref{sec_group_test} (see Corollary~\ref{corr_suff_gt_indirect}).
Since $C_0 \le 1$, this bound
is tighter than the bound obtained by using the same number of observations 
as is required to find the active set \cite{Atia_BooleanCS}. 


\subsection{Decoding Scheme 2: Find the Inactive Subset Directly}
\label{DecSch2}
For simplicity of exposition, we describe
this decoding scheme in two stages:
First, we present the result for the $K=1$ case, i.e., when there is only 
one active variable.
This case brings out
the fundamental difference between finding active and inactive variables.
We then generalize our decoding scheme to $K > 1$.
\subsubsection{The $K=1$ Case} \label{sec_K1case}
We start by proposing the following decoding scheme:
\begin{itemize}
  \item Given $\{\matX, \uy\}$, compute $P(\uy | \ux_i)$ for all $i \in [N]$ 
   and sort them in descending order. Since $K=1$, we know 
   $P(Y|X_i)$ for all $i \in [N]$, and
   hence $P(\uy | \ux_i)$ can be computed using the independence assumption across 
   different observations.
 \item Pick the last $L$ indices in the sorted array as the set of $L$ inactive variables.
\end{itemize}
Note that, in contrast to finding active set, the problem of finding $L$ inactive variables 
does not have unique solution (except for $L=N-K$). 
The proposed decoding scheme provides a way to pick a solution, and the probability 
of error analysis take into account the fact that an error event happens only when the 
inactive set chosen by the decoding algorithm contains an active variable. 
\begin{theorem} \label{Suff_K1Thm}
Let $N$, $M$, $L$ and $K$ be as defined above. Let $K=1$. 
Let $E_0$ and $I^{(j)}$ be as defined in (\ref{eq:E0_def}) and (\ref{eq:MI_gen_def}).
Let $\rho \in [0~ 1]$.
With the above decoding scheme, the average probability of error, $P_e$, in 
finding $L$ inactive variables is upper bounded as
\begin{align} \label{eq:pe_k1case}
  P_e \le \exp \left [ - \left ( M E_0(\rho, 1, N-L) - {\rho \log {N-1 \choose L-1}}
   \right ) \right ].
\end{align}
Further, for any $\epsilon_0 > 0$, if
\begin{align} \label{eq:suff_cond1}
  M &> (1 + \epsilon_0) \frac{\log {N - 1 \choose L-1} }{(N - L) I^{(1)}},
\end{align}
then there exists $\epsilon_1 > 0$, independent of $N$ and $L$,  such that 
$P_e \le \exp\lb -\epsilon_1 \log {N-1 \choose L-1} \rb$. 

\noindent \emph{Proof:} See Sec. \ref{sec_proof_suffkeq1}. 

\end{theorem}
\noindent We make the following observations:
\begin{enumerate}[(a)]
\item \label{k1_1} 
  Figure~\ref{figure:suff_bnd_compare} compares
  the above bound on the number of observations with the bounds for the decoding
  scheme presented in  Section~\ref{DecSch1}\footnote{We refer the reader to the remark at the end of 
    the proof for Theorem~\ref{Suff_K1Thm} (Section~\ref{sec_proof_suffkgt1}) for 
    a bound on the sufficient number of observations, resulting from Corollary~\ref{corr1_indir},
  corresponding to $K=1$ case.} and in Theorem III.I\cite{Atia_BooleanCS}, 
  for the $K=1$ case.
  \item \label{k1_2} 
  Consider the case $L=N-1$, i.e., we want to find all the inactive variables.
  This task is equivalent to finding the active variable.
  The above decoding scheme for finding $N-1$ inactive variables is
  equivalent
  \footnote{The decoding schemes are equivalent in the sense that an  error in finding $K$ active variables implies an error in finding $N-K$ inactive variables, and vice-versa.}
  to the maximum likelihood criterion based decoding scheme 
  used in Theorem III.I in \cite{Atia_BooleanCS} 
  for finding $1$ active variable. 
  This is also reflected in the above result, as the number of observations 
  sufficient for finding $N-1$ inactive variables matches exactly with the 
  number of observations sufficient for finding $1$ active variable 
  (see Figure~\ref{figure:suff_bnd_compare}).
\end{enumerate}
%
\subsubsection{$K > 1$ Case} \label{sec_kgt1_case}
For $K > 1$, by arranging $P(\uy | \matX_{S_i})$ 
in decreasing order for all $S_i \subset [N]$ such that $|S_i| = K$, it is possible for
the sets $S_i$ towards the end of the sorted list to have overlapping entries.
Thus, in this case the decoding algorithm proceeds by picking up just the sufficient
number of $K$-sized sets from the end that provides us with a set of $L$ inactive variables.
We propose the following decoding scheme:


\emph{Decoding Scheme}:
%

\begin{enumerate}
  \item Given $\{\matX, \uy\}$, compute $P(\uy | \matX_{S_i})$ for all 
    $S_{i} \subset [N]$ such that $|S_{i}| = K$, and  sort these in descending order. Let the ordering
    be denoted by $\{S_{i_1}, S_{i_2}, \ldots, S_{i_{N \choose K}}\}$.
  \item Choose $n_0$ sets from the end such that
    \begin{equation} \label{eq:minimum_llr_sets}
      \vert \mathop{\bigcup}_{l=1}^{n_0} S_{i_{ {N \choose K} - l + 1}} \vert \ge L ~ ~ ~ \text{and} ~ ~ ~
      \vert \mathop{\bigcup}_{l=1}^{n_0-1} S_{i_{ {N \choose K} - l + 1}} \vert <  L.
    \end{equation}
  \item    Let $\Omega_{\text{last}} \triangleq \{i_{N \choose K}, i_{ {N \choose K}-1}, \ldots, i_{ {N \choose K}-n_0+1} \}$ 
    denote this set of last $n_0$ indices. Declare 
    $S_H \triangleq {\bigcup}_{j \in \Omega_{\text{last}}} S_{j} $ as the decoded set of inactive variables.
\end{enumerate}
That is, choose the minimum number of $K$-sized sets with least likelihoods such that
we get $L$ distinct variables and declare these as the decoded set of inactive variables.
We refer to this decoding scheme as the ``direct'' decoding scheme.
We note that $S_H$ might contain more than $L$ items. In particular, $L \le |S_H| \le L + K -1$. 
Further, for all values of $L$ such that $L < (N-K)-(K-1)$, the complement set of $S_H$, i.e., 
$[N]\backslash S_H$, will contain at least $L_0 \triangleq (N - L - 2K + 1)$ variables from 
inactive set ($[N]\backslash S_1$). This fact will be useful in achieving an upper bound on the decoding error probability 
for this algorithm.
We summarize the probability of error analysis of the above algorithm in
the following theorem. 
\begin{theorem} \label{Suff_Kgt1Thm1}
Let $N$, $M$, $L$ and $K$ be as defined above. 
Let $L_0 \triangleq (N - L - 2K + 1)$. 
For any $\rho \in [0~1]$ and
any $1 \le L < (N-K)-(K-1)$, with the above decoding scheme, the average 
probability of error, $P_e$, in finding $L$ inactive variables is upper bounded as
\begin{align} \label{eq:nonasym_pe}
 P_e  & \le \exp \left [ - \left \{ M E_0(\rho,1, L_0) - { \rho \log {N-K \choose L_0} }- 
   { \log \left[K {N - 1 - L_0 \choose K -1} \right]} \right \} \right ].
\end{align}
\end{theorem}
\noindent \emph{Proof:} See Sec.~\ref{sec_proof_suffkgt1}.

The above result is applicable to the abstract sparse signal model 
specified in Section~\ref{sec:PrbSetup}. 
It can be specialized to any particular sparse signal model, for example, that of non-adaptive group testing,  by lower bounding $E_0(\rho, 1, L_0)$, 
to obtain a relationship between $M$ and the average probability of error for the decoding algorithm.
We present the results for the case of the non-adaptive group testing in 
Section~\ref{sec_group_test}.

\begin{figure}[t]
\centering
\includegraphics[scale=0.7]{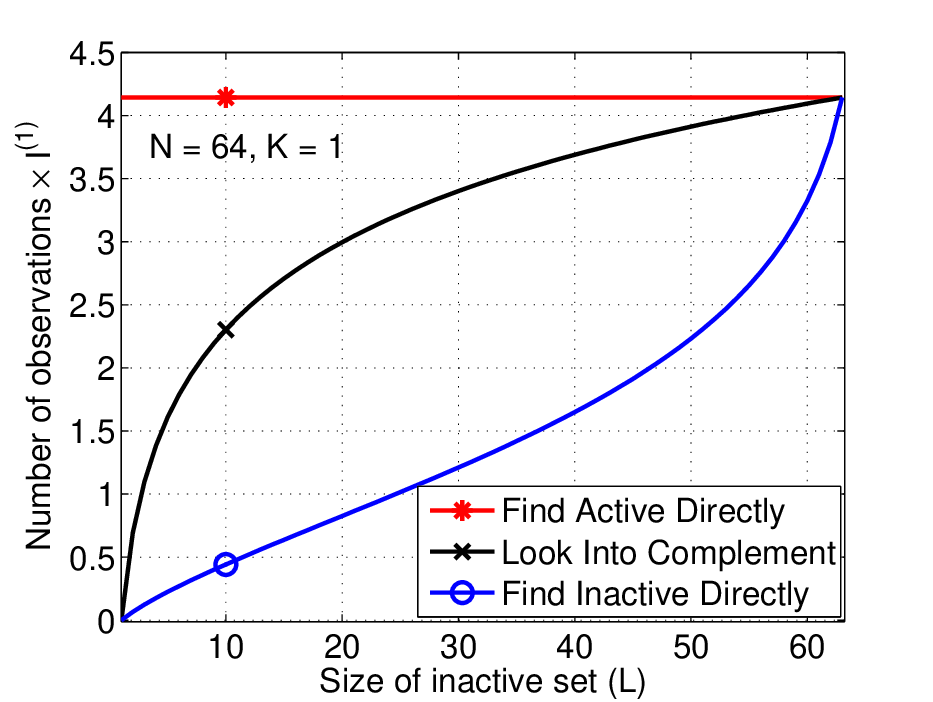}
\caption{Sufficiency bounds on the number of observations required to find $L$ 
  inactive variables for $K=1$ case.
The comparison  is presented with respect to the value of $M I^{(1)}$, as the
    application-dependent mutual information term $ I^{(1)}$ is common to all the bounds.
  The approach of finding the $L$ inactive variables directly, 
  especially for small values of $L$, requires
  significantly fewer number of observations compared to the approach of finding
  the inactive variables indirectly, after first identifying the active variables. 
The plot corresponding to the curve labeled \texttt{Find Active Directly} refers to the number of observations that are sufficient for finding the $K$ active variables\cite{Atia_BooleanCS}.}
\label{figure:suff_bnd_compare}
\end{figure}

\section{Necessary Number of Observations} \label{necc_cond}
In this section, we derive lower bounds on the number of observations
required to find a set of $L$ inactive variables, in the sense that if
the number of observations is lower than the bound, the probability of error will be bounded strictly away from
zero, regardless of the decoding algorithm used. Here, we need to lower bound the probability of error in choosing a set
of $L$ inactive variables. To this end, we employ an adaptation of 
Fano's inequality \cite{Gallager_bk,Cover_bk}. 

Let $\mathcal{I}^d \triangleq \left \{ \omega_1, \omega_2, \ldots, \omega_{N \choose K} 
\right \}$ be the collection of all $K$ sized subsets of $[N]$ such that $| S_{\omega_i}| = K$ for 
$i=1, 2, \ldots, {N \choose K}$. For each $\omega \in \mathcal{I}^d$ let us
associate a collection of sets, $\mathcal{I}^h_{\omega}\triangleq \left \{ \alpha_1, \alpha_2, \ldots, \alpha_{N-K \choose L} 
\right \}$, such that $| S_{\alpha_i} | = L $ and $
S_{\alpha_i} \cap S_{\omega} = \{ 0 \}$, $i=1, 2, \ldots , {N-K \choose L}$.
That is, $\mathcal{I}^h_{\omega}$ is the collection of all $L$-sized 
subsets of all-inactive variables when $S_\omega$ represents the active
set. 
Also, let $\mathcal{I}^H$ denote the set of all $L$-sized subsets of $[N]$.
Note that $|\mathcal{I}^H| = {N \choose L}$.
Given the observation vector, $\uy  \in \mc{Y}^M$, let 
$\phi:\mc{Y}^M \rightarrow \mathcal{I}^H$ 
denote a decoding function, such that $\hat{\alpha} = \phi(\uy)$ is the decoded
set of $L$ inactive variables.
Given an active set $\omega$ and an observation vector $\uy$, an error occurs if
$\hat{\alpha} \notin \mathcal{I}^h_\omega$. Define,
\begin{align} \label{eq:lb_pe_def}
P_e = P(\hat{\alpha} \notin \mathcal{I}^h_\omega).
\end{align}
Define a binary error RV, $E$, as
$E \triangleq \mathbb{I}_{\left \{ \hat{\alpha} \notin \mathcal{I}^h_{\omega} \right \} }$.
Note that $P_e = \text{Pr}(E = 1)$.
We state a necessary condition on the number of observations in the
following theorem.
\begin{theorem} \label{thm_necc_cond}
  Let $N$, $M$, $L$ and $K$ be as defined before. 
  Let $I^{(j)}$ and $P_e$ be as defined in (\ref{eq:MI_gen_def}) and
  (\ref{eq:lb_pe_def}), respectively. 
  A necessary condition on the number of observations $M$ required to 
  find $L$ inactive variables with asymptotically vanishing probability 
  of error, i.e., 
  $\lim_{N \rightarrow \infty} P_e = 0$, is given by
\begin{align} \label{eq:necc_cond_M}
  M &\ge \max_{1 \le j \le K} \frac{\Gamma_l(L, N, K, j)}
  {I^{(j)}}(1 - \eta), ~ ~ \mbox{where} ~ ~ ~
  \Gamma_l(L, N, K, j)  \triangleq \log \left[ \frac{ {N- K + j \choose j}}
  {{ N- K + j - L \choose j}} \right ],
\end{align}
for some $\eta > 0$.
\end{theorem}
\noindent \emph{Proof:} See Sec.~\ref{sec_proof_necc}. 

That is, any sequence of random codebooks, that achieves 
$\lim_{N \rightarrow \infty} P_e = 0$, must satisfy 
the above bound on the length of the codewords.
Given a specific application, we can bound $I^{(j)}$ for each $j=1,2, \ldots, K$, and obtain a
characterization on the necessary number of observations, as we show in the next section.


\section{Finding Non-Defective Items Via Group Testing} \label{sec_group_test}
In this section, we apply the above mutual information based 
results to the specific case of non-adaptive group testing, and characterize
the number of tests to identify a subset of non-defective items in a large population.
We consider a random pooling based, noisy non-adaptive group 
testing framework\cite{Atia_BooleanCS, du2006}, as described in detail in
Section~\ref{sec:PrbSetup}.
Our goal here is to find upper and lower bounds on the number of tests 
required to identify an $L$ sized subset belonging to $[N] \backslash \mathcal{G}$ 
using the observations $\underline{y}$, with vanishing probability of 
error as $N \rightarrow \infty$.
We focus on the regime where $K, L, N \rightarrow \infty$ with 
    $\frac{L}{N-K} \rightarrow \alpha_0$, $\frac{K}{N-K} \rightarrow \beta_0$ for some 
    fixed $\alpha_0, \beta_0 \in (0,1)$. 

To compute the lower bounds on the number of tests, using the results of
Theorem~\ref{thm_necc_cond}, we need to upper bound the mutual information term, $I^{(j)}$, for the 
group testing signal model given in (\ref{eq:gtmodel}).
Using the bounds on $I^{(j)}$\cite{johnson_gtbp}, with\footnote{In general,
$p = \frac{\alpha}{K}$, with $\alpha$ depending upon $u$ and $q$, 
is useful for bounding the mutual information 
terms $I^{(j)}$~\cite{johnson_gtbp, Atia_BooleanCS}.}
$p=\frac{1}{K}$ and $u \le 0.5$,
we summarize the order-accurate lower bounds on the number of tests to 
find a set of $L$ non-defective items in Table \ref{tab:tab_necc_order}.
A brief sketch of the derivation of these results is provided in
Appendix~\ref{order_res_proof}.

To compute the upper bounds on the number of tests, we need to lower bound $E_0(\rho,1,n)$ for some $\rho \in [0,1]$ 
and show that the negative exponent in the probability of error term in (\ref{eq:nonasym_pe}) can be made strictly greater 
than~$0$ by choosing $M$ sufficiently large. We first present the lower bounds on $E_0(\rho,1,n)$ in the following
lemma.
\begin{lemma} \label{E0LB_Thm}
Let $N$, $M$, $L$ and $K$ be as defined above. 
Let $L_0 = (N - L - 2K + 1)$. 
Let $E_0(\rho, j, n)$ be as defined in (\ref{eq:E0_def}) and 
define $\rho_0 \triangleq \frac{K-1}{L_0}$.
For the non-adaptive group testing model with $p=\frac{1}{K}$ and
for all values of $L \le (N-3K+1)$, we have
\begin{enumerate}[(a)]
  \item For the noiseless case ($u=0, q=0$):
\begin{align} \label{eq:e0_lb_nonoise}
  E_0(\rho_0, 1, L_0) \ge \frac{(1-e^{-1}) - (\frac{1}{K})^K}{e}
\end{align}
  \item For the additive noise only case ($u=0, q>0$):
    \begin{align} \label{eq:e0_lb_q}
      E_0(\rho_0, 1, L_0) \ge \frac{e^{-2}}{4}(1-q)
\end{align}
  \item For the dilution noise only case ($u>0, q=0$):
\begin{align} \label{eq:e0_lb_u}
  E_0(\rho_0, 1, L_0) \ge \frac{e^{-2}}{4}(1-u^{\frac{1}{K}})
\end{align}
\end{enumerate}
\end{lemma}
The proof of the above lemma is presented in Appendix~\ref{sec_proof_E0}.
For notational convenience, we 
let $E_0^{(lb)}$ denote a common lower bound on $E_0(\rho_0, 1, L_0)$, as derived above.
The following lemma presents an upper bound on the number of tests 
required to identify $L$ non-defective items in a non-adaptive group testing setup.
\begin{theorem} \label{lemma_suff_NN_KON}
  Let $P_e$ be the average probability of error in finding $L$ inactive variables
  under the decoding scheme described in Section~\ref{sec_kgt1_case}, which is
  upper bounded by (\ref{eq:nonasym_pe}).
  Let $L_0 \triangleq (N - L - 2K + 1)$ and let $\theta_0 \triangleq \frac{L+K-1}{N-K}$. 
  Then, for any $\epsilon_0 > 0$ and all values of $L \le (N-3K+1)$,
  if $M$ is chosen as
  \begin{align} \label{eq:suff_NN_KON}
    M > (1 + \epsilon_0) \frac{K-1}{E_0^{(lb)}}\left [
	    \frac{ H_b(\theta_0) }{1 - \theta_0} +
	    \log \left( 2 + \frac{L}{K-1}\right)  + 1 + 
	    \frac{ \log K }{K-1} \right],
  \end{align}
  then $P_e \le \exp\left(- \epsilon_0 (K-1) \log \frac{N-K}{L_0} \right)$.
\end{theorem}
An outline of the proof is presented in Section~\ref{proof_suff_NN_KON}.
In the regime where $L, K \rightarrow \infty$ 
as $N \rightarrow \infty$, it follows from the above lemma that $\lim_{N \rightarrow \infty} P_e = 0$.

Finally, we present an upper bound on the number of tests obtained for the indirect 
decoding scheme presented in Section~\ref{DecSch1} for the noiseless case.
Using~\cite[Lemma VII.$1$ and VII.$3$]{Atia_BooleanCS} to lower bound $E_0(\rho, j, 1)$ for
the noiseless case, and noting that, from the union bound, we have
$C_0(j) \le \frac{ j {N-K-1 \choose L-1 }}{ {N-K \choose L}} = j \frac{L}{N-K}$,
the following corollary builds on the results presented in Corollary~\ref{corr1_indir}.
\begin{corollary} \label{corr_suff_gt_indirect}
  Let $P_e$ be the average probability of error in finding $L$ inactive variables
  under the decoding scheme described in Section~\ref{DecSch1}, which is
  upper bounded by (\ref{eq:nonasym_pe_cond0}). 
  For any $\epsilon_0 > 0$, there exists absolute constant $c_0>2$, independent of $N$, $K$ and $L$, such 
  that if $M$ is chosen as
  \begin{align} \label{eq:suff_NN_KON_indir}
    M > (1+\epsilon_0) c_0 K \log L \log^2 K,
  \end{align}
  then $P_e \le \exp\left(-\epsilon_0 (K \log L) \right)$.
\end{corollary}
    We now make following observations about the results presented in this section.
    We consider the regime where for some fixed $\alpha_0, \beta_0 \in (0,1)$,
    $\frac{L}{N-K} \rightarrow \alpha_0$, $\frac{K}{N-K} \rightarrow \beta_0$ as $N \rightarrow \infty$.
    Further, as we are typically interested in $L \ge K$ and, also since our results apply 
    for $L \le N - 3K + 1$, we only consider $\alpha_0$, $\beta_0$ such that 
    $\beta_0 \le \alpha_0$ and $\alpha_0 + 2 \beta_0 \le 1$. For the indirect decoding scheme presented in Section~\ref{DecSch2}, we summarize the upper bounds on 
    the number of tests to find a set of $L$ non-defective items in Table~\ref{tab:tab_suff_order}.
\begin{enumerate}[(a)]
  \item 
    
    We first consider the noiseless case.
    \begin{enumerate}[(i)]
      \item For the direct decoding scheme, $O(K)$ number of tests are sufficient.
    In comparison, using results from Corollary~\ref{corr_suff_gt_indirect},    $O(K \log L \log^2 K)$ tests are sufficient for the indirect decoding scheme.
   Also, from \cite[Theorem~V.2]{Atia_BooleanCS}, 
    $O(K \log N \log^2 K)$ tests are sufficient for finding all the defective items.
    Thus, in this case, the direct decoding scheme for finding non-defective items performs better 
    compared to the indirect decoding schemes by a poly-log factor of the number of defective items, $K$.
    Further, from Table~\ref{tab:tab_necc_order}, we observe that the upper bound
    on the number of tests for the direct decoding scheme is within a $c \log K$ factor of the lower bound, where
    $c$ is a constant independent of $N$, $L$ and $K$.
  \item  The size of non-defective set, $L$, impacts the upper bound on the number of 
 tests only through $\alpha_0$, i.e., the fraction of non-defective items that need to be found. 
 From Table~\ref{tab:tab_suff_order},  $\Phi(\alpha_0, \beta_0)$ is an increasing function of 
 $\alpha_0$. That is, a higher $\alpha_0$ results in a higher rate at which the upper bound on the number
 of tests increases with $K$.

    \end{enumerate}

  \item Performance under noisy observations: 
    \begin{enumerate}[(i)]
      \item 
	For the additive noise, $O(\frac{K}{1-q})$ number of tests are sufficient for the direct 
	decoding scheme 
	The indirect scheme (as well as finding the defective items) also show similar $\frac{1}{1-q}$ factor
	increase in the number of tests under additive noise scenario (see, e.g., \cite[Theorem VI.2]{Atia_BooleanCS}). 
    Further, from Table~\ref{tab:tab_necc_order}, we observe that for fixed $\alpha_0, \beta_0$ and $q$, the upper bound
    on the number of tests for the direct decoding scheme is within a constant factor of the lower bound.
  \item 
    For dilution noise, $O\lb \frac{K}{1-u^\frac{1}{K}}\rb$ are sufficient for the direct
    decoding scheme.
    Another characterization for the sufficient number of tests for the direct decoding scheme, 
    based on the remark at the end of Appendix~\ref{sec_proof_E0}, is $O\lb \frac{K^2}{1-u^\frac{1}{2}}\rb$ 
    number of tests.
The direct decoding scheme shows high sensitivity to the dilution noise. This
    behavior is in sharp contrast to the indirect scheme, where the dilution noise parameter $u$ leads to an
    increase in the number of tests only by a factor of $\frac{c}{1-u}$
    (see, e.g., \cite[Theorem VI.5]{Atia_BooleanCS}). From Table~\ref{tab:tab_necc_order}, 
    the lower bounds also show an increase in the number of tests by a factor $\frac{1}{1-u}$ for the
    dilution noise scenario.
    The conservativeness of the upper bound for the direct decoding scheme in the presence of dilution noise appears 
    to be due to the following factors: (a) The lower bound on $E_0$ is $\Omega(\frac{1}{K})$, which underscores the general fact
    that the group testing system is more sensitive to the diluton noise, and (b) The term $ \log {N-1-L_0 \choose K-1}$ 
    in (\ref{eq:nonasym_pe}), which might be due to the particular decoding scheme employed or the
    specific technique employed in bounding the error exponent.
\end{enumerate}
\end{enumerate}

\begin{table}[t]
\centering
\caption{Finding a subset of $L$ non-defective items: Results for necessary number 
of group tests which hold asymptotically as $(N,K,L) \rightarrow \infty$,
$\frac{L}{N-K} \rightarrow \alpha_0$  with $0 < \alpha_0 < 1$. The constants $C_n, C_n', C_n'' > 0$ are
independent of $N, L, K, u$ and $q$.}
\begin{tabular}{|l|l|} \hline 
  No Noise    & $\dfrac{C_n K}{\log K}\log \dfrac{1}{[1 - \alpha_0 + o(1)]}  $ \\ \hline
    Additive Noise    & $\dfrac{C_n' K}{\log \frac{1}{q}}\log \dfrac{1}{[1 - \alpha_0 + o(1)]} $ \\ \hline
    Dilution Noise    & $\dfrac{C_n'' K}{(1-u)\log K }\log \dfrac{1}{[1 - \alpha_0 + o(1)]} $ \\ \hline
\end{tabular}
\label{tab:tab_necc_order}
\end{table}


\begin{table}[t]
\centering
\caption{Finding a subset of $L$ non-defective items: Results for sufficient number 
of group tests which hold asymptotically as $(N,K,L) \rightarrow \infty$,
$\frac{L}{N-K} \rightarrow \alpha_0$ and $\frac{K}{N-K} \rightarrow \beta_0$  with $0 < \beta_0 \le \alpha_0 < 1$ such
that $\alpha_0 + 2 \beta_0 < 1$.
Define $\Phi(\alpha_0, \beta_0) \triangleq \lb \frac{H_b(\gamma_0)}{1 - \gamma_0} + \log \frac{\alpha_0}{\beta_0} \rb$,
where $\gamma_0 = \alpha_0 + \beta_0$. The constants $C_s, C_s', C_s'' > 0$ are
independent of $N, L, K, u$ and $q$.}
\begin{tabular}{|l|l|} \hline 
  No Noise & $ \dfrac{C_s K}{(1 - o(1))} \ls \Phi(\alpha_0, \beta_0) + o(1)\rs  $ \\ \hline
  Additive Noise & $ \dfrac{C_s' K}{(1-q)} \ls \Phi(\alpha_0, \beta_0) + o(1)  \rs  $ \\ \hline
  Dilution Noise & $ \dfrac{C_s'' K}{(1-u^{\frac{1}{K}})} \ls \Phi(\alpha_0, \beta_0) + o(1)\rs  $ \\ \hline
\end{tabular}
\label{tab:tab_suff_order}
\end{table}

\section{Proofs of the Main Results} \label{sec_thm_proofs}
We now present the proofs of Theorems \ref{Suff_K1Thm}, \ref{Suff_Kgt1Thm1} and \ref{thm_necc_cond}, which constitute the main results in this paper. 
\subsection{Proof of Theorem \ref{Suff_K1Thm}: Sufficient Number of Observations, $K=1$ Case}
\label{sec_proof_suffkeq1}
At the heart of the proof of this theorem is the derivation of
an upper bound on the average
probability of error in finding $L$ inactive variables using the decoding 
scheme described in Section~\ref{sec_K1case}. In turn, the upper bound is obtained by characterizing the error 
exponents on the average probability of error~\cite{Gallager_bk}. 
Without loss of generality, due to the symmetry in the model, we can assume that the RV
$X_1$ is active.
Given that $X_1$ is the active variable, the decoding algorithm 
will make an error if $P(\uy | \matX_1)$ falls within the last 
$L$ entries of the sorted array generated as described in the
decoding scheme.
Let $\uy$ be the observed output, and let $\mc{E}$ denote the event that an error has occurred, 
when item $X_1$ is the active variable and $\matX_1$ is the first column of $\matX$.
Further, let $\text{Pr}(\mc{E})$ be a shorthand for  
$\text{Pr}\{ \text{error}| X_1\text{ is active}, \matX_1, \uy\}$.
The overall average probability of error, $P_e$, can be expressed as
\begin{align} \label{eq:k1_pe_1}
  P_{e} &= \sum_{\uy, \matX_1}  P(\uy | \matX_1) Q(\matX_1) 
  \text{Pr}(\mc{E}).
\end{align}
Let $S_z \subset [N] \backslash 1$ be a set of $N-L$ items, i.e., $|S_z| = N-L$.
Let $\mathbf{\mc{S}_z}$ denote the set of all possible $S_z$.
Further, let $\mc{A}_{S_z} \subset \{\matX_{S_z}\}$ be such that, 
$\mc{A}_{S_z} = \{ \{\matX_{S_z}\}: P(\uy|\matX_j) ~ \geq ~ P(\uy|\matX_1) 
~ \forall ~ j \in S_{z} \}$. 
That is, $\mc{A}_{S_z}$ represents all those realizations of the
random variable $\matX_{S_z}$ which satisfy the above condition, which
states that each variable in $S_z$ is more likely than the active variable, $X_1$.
It is easy to see that
$\mc{E} \subset \mc{A} \triangleq \bigcup_{{S_z} \in \mc{S}_z } \mc{A}_{S_z}$, i.e., an error event implies
that there exists at least one set of $N-L$ variables, $S_z$, such that
$P(\uy|\matX_j) ~ \geq ~ P(\uy|\matX_1) ~ \forall ~ j \in S_{z}$.
Thus, $\text{Pr}(\mc{E}) \le \text{Pr}(\mc{A} )$. 
Let $s$ be an optimization variable such that $0 \le s \le 1$.
The following set of inequalities upper bound $\text{Pr}(\mc{E})$:
\restarthintedrel
\begin{align}
  \nonumber
  \text{Pr}(\mc{E}) &\le \sum_{S_z \in \mathbf{\mc{S}_z}} ~
  \sum_{ \matX_{S_z} \in\mc{A}_{S_z}} Q(\matX_{S_z}) \\
  \nonumber
  &\hintedrel[in1]{\le}\sum_{S_z \in \mathbf{\mc{S}_z}} ~ 
  \sum_{ \matX_{S_z} \in\mc{A}_{S_z}} Q(\matX_{S_z})
  \prod_{j \in S_z} \left [ \frac{P(\uy|\matX_j)}{P(\uy|\matX_1)} \right ]^s \\
  \nonumber
  &\hintedrel[in2]{\le} \sum_{S_z \in \mathbf{\mc{S}_z} } ~
  \sum_{\matX_{S_z} } ~ 
  \prod_{j \in S_z} Q(\matX_j) \left [ \frac{P(\uy|\matX_j)}{P(\uy|\matX_1)} \right ]^s \\
  \nonumber
  &\hintedrel[in3]{=} \sum_{S_z \in \mathbf{\mc{S}_z} } ~
  \prod_{j \in S_z} ~ \sum_{\matX_j} Q(\matX_j) \left [ \frac{P(\uy|\matX_j)}{P(\uy|\matX_1)} 
    \right ]^s \\
   \label{eq:k1uppbndq}
    &\hintedrel[in4]{=} {N-1 \choose L-1}
  \left \{ \sum_{\matX_j} Q(\matX_j) \left [ \frac{P(\uy|\matX_j)}{P(\uy|\matX_1)} 
  \right ]^s \right \}^{N-L}.
\end{align}
In the above, (\hintref{in1}) follows since we are multiplying 
with terms that are all greater than $1$ and  
$(\hintref{in2})$ follows since we are adding extra nonnegative terms by summing over all $\matX_{S_z}$. 
$(\hintref{in3})$ follows by using the independence of the codewords, i.e.,  $Q(\matX_{S_z}) = \prod_{j \in S_z} Q(\matX_j)$, 
and simplifying further.
$(\hintref{in4})$ follows since the value of the expression inside the product term does not depend upon
any particular~$j$.

Let $0 \le \rho \le 1$.  If the R.H.S.\ in (\ref{eq:k1uppbndq}) is less than
$1$, then raising it to the power $\rho$ makes it bigger, and if it is 
greater than $1$, it remains greater than $1$ after raising it to the power $\rho$.
Thus, we get the following upper bound on $\text{Pr}(\mc{E})$:\footnote{This is a standard 
Gallager bounding technique~\cite[Section~5.6]{Gallager_bk}.}
\begin{align} \label{eq:k1rhofinal}
  \text{Pr}(\mc{E}) &\le {N-1 \choose L-1}^{\rho}
  \left \{ \sum_{\matX_j} Q(\matX_j) \left [ \frac{P(\uy|\matX_j)}{P(\uy|\matX_1)} 
  \right ]^s \right \}^{\rho(N-L)}.
\end{align}
Substituting this into (\ref{eq:k1_pe_1}) and simplifying, we get
\begin{align} \label{eq:k1_opts_eqn}
  P_e \le {N-1 \choose L-1}^{\rho} \sum_{\uy} 
  \sum_{\matX_1} Q(\matX_1) P(\uy | \matX_1)^{1 - \rho(N-L)s}
  \left \{ \sum_{\matX_j} Q(\matX_j) P(\uy|\matX_j)^s
  \right \}^{\rho(N-L)}.
\end{align}
Putting $s = 1/(1 + \rho(N-L))$, we get
\begin{align} \label{eq:k1_opts_eqn_1}
  P_e \le {N-1 \choose L-1}^{\rho} \sum_{\uy} 
  \left \{ \sum_{\matX_j} Q(\matX_j) P(\uy|\matX_j)^{\frac{1}{1 + \rho(N-L)}}
  \right \}^{1 + \rho(N-L)}.
\end{align}
Finally, using the independence across observations and using the definition of 
$E_0 (\rho, j, n)$ from (\ref{eq:E0_def}) with $j=1$ and $n = N-L$, we get
\begin{align} \label{eq:k1_opts_eqn_2}
  \nonumber
  P_e &\le {N-1 \choose L-1}^{\rho} \left [ \sum_{Y \in \mathcal{Y}} 
    \left \{ \sum_{X_j \in \mathcal{X}} Q(X_j) P(Y|X_j)^{\frac{1}{1 + \rho(N-L)}}
\right \}^{1 + \rho(N-L)} \right]^M \\
&= \exp^{\left [ -M F(\rho) \right]}, ~ \mbox{where} ~ 
F(\rho) = E_0(\rho,1,N-L) - \frac{\rho \log {N-1 \choose L-1}}{M}.
\end{align}
Hence (\ref{eq:pe_k1case}) follows. 

For the following discussion, we treat $F$ and $E_0$ as functions of $\rho$ only
and all the derivatives are with respect to $\rho$.
Note that $ F^{'}(\rho) = E_0^{'}(\rho) - \frac{ \log {N-1 \choose L-1}}{M}$.
It is easy to see that $E_0(0) = 0$ and hence $F(0) = 0$. 
With some calculation, we get,
\begin{align}
  \left. E_0^{'}(\rho) \right|_{\rho = 0} = (N-L) \sum_{Y,X} P(Y,X) \log \frac{P(Y|X)}{P(Y)} = (N-L) I^{(1)}.
\end{align}
Using the Taylor series expansion of $E_0(\rho)$, and following similar
analysis as in \cite[Section~III.D]{Atia_BooleanCS}, it is easy to show that there exists a 
$\rho \in (0,1]$, sufficiently small, such that if $M$ is 
chosen as in (\ref{eq:suff_cond1}), then $M F(\rho) > \epsilon_1 (N-L) \log {N-1 \choose L-1}$ for some $\epsilon_1 > 0$,
independent of $N$ and $L$.
This completes the proof.

\emph{Remark:} For the decoding scheme described in \ref{DecSch1}, for the $K=1$ case,
using similar arguments as the above, if $M > (1+\epsilon_0) \frac{\log (L-1)}{I^{(1)}}$
for any $\epsilon_0 > 0$, then
there exists $\epsilon_1 > 0$, and independent of $N$ and $L$,  such 
that $P_e \le \exp (- \epsilon_1 \log L)$, i.e., $P_e \rightarrow 0$, as
$L \rightarrow \infty$.

\subsection{Proof of Theorem \ref{Suff_Kgt1Thm1}: Sufficient Number of Observations, $K > 1$ Case}
\label{sec_proof_suffkgt1}
The decoding algorithm outputs a set, $S_H$, of  at least $L$ 
inactive variables. A decoding error happens if the set $S_H$ contains one or more variables from
the active set. We now upper bound the average probability of error of the proposed decoding algorithm.
The probability is averaged over all possible instantiations of $\{\matX, \uy \}$ as 
well as over all possible active sets.
By symmetry of the codebook ($\matX$) construction, the average probability of error is the same for 
all the active sets.  Hence, we fix the active set and then compute average probability of error 
with this set. Let $S_1 \subset [N]$ be the active set such that $|S_1| = K$.
We also define the following notation: 
For any set $S_{\omega} \subset [N]$ such that $|S_{\omega}| = K$ and for any item $j \in S_{\omega}$,
let $S_{\omega j^c} \triangleq S_{\omega} \backslash j$. Note that $|S_{\omega j^c } | = K-1$.

For any $d \in S_1$, define $\mc{E}_d$ to be the error event such that $d$ belongs to $S_H$.
The overall average probability of error, $P_e$, in finding $L$ inactive
variables can thus be upper bounded as
\begin{align} \label{eq:overallpe_eq}
  P_e \le \sum_{d \in S_1} \text{Pr}({\mc{E}_d}). 
\end{align}
Further, 
\begin{equation} \label{eq:pe_ed}
\text{Pr}({\mc{E}_d}) = \sum_{\uy} \sum_{\matX_{S_{1}}}
  P(\uy | \matX_{S_1})  Q(\matX_{S_1})  \left[ 
  \text{Pr} \{ \mc{E}_d | S_1~ \mbox{is the active set}, \uy, \matX_{S_1} \} \right].
\end{equation}

We now upper bound $\text{Pr} \{ \mc{E}_d | S_1~ \mbox{is the active set}, \uy, \matX_{S_1} \}$.
Let $S_z \subset [N] \backslash S_1$ be such that $|S_z| = L_0$.
Let $S_{\omega} \subset [N]$ be a $K$ sized index set such that $S_\omega = \{d \cup S_{\omega d^c} \}$, where 
$S_{\omega d^c} \subset [N]\backslash \{d\} \backslash S_z$ and $d \in S_1$.
Further, let $\mathbf{\mc{S}_z}$ and $\mathbf{\mc{S}_{\omega d^c}}$ be the collection of
all possible sets $S_z$ and $S_{\omega d^c}$, respectively. 
It is easy to see that $|\mathbf{\mc{S}_z}| = {N-K \choose L_0}$ and 
$|\mathbf{\mc{S}_{\omega d^c}}| = {N-1 - L_0 \choose K-1}$.
With $S_1$ as the active set, $d \in S_1$, the observed output $\uy$ and the codebook entries corresponding 
to set $S_1$ as $\matX_{S_1}$, 
define $\mc{A}_d({S_z, S_{\omega d^c}}) \subset \{\matX_{S_{z} \cup S_{\omega d^c}}\}$ and $\mc{A}_d$ as follows:
\begin{align} \label{eq:Adef}
  \mc{A}_d({S_z, S_{\omega d^c}}) &= \{ \{\matX_{S_{z}}, \matX_{S_{\omega d^c}}\} : ~
  P(\uy|\matX_{\alpha}, \matX_{S_{\omega d^c}}) \geq 
P(\uy|\matX_d, \matX_{S_{\omega d^c}}) ~ 
\forall ~ {\alpha} \in S_{z}  \}, \\
\mc{A}_d &= \bigcup_{S_z \in \mc{S}_z} \bigcup_{S_{\omega d^c} \in \mc{S}_{\omega d^c}} 
\mc{A}_d({S_z, S_{\omega d^c}}).
\end{align}
That is, $\mc{A}_d(S_z, S_{\omega d^c})$ represents a set of the those realizations of the random variables
$\matX_{S_z}$ and $\matX_{S_{\omega d^c}}$ which satisfy the condition in (\ref{eq:Adef}).
\begin{proposition} \label{prop_err_up_bnd}
  $\text{Pr} \{ \mc{E}_d | S_1~ \mbox{is the active set}, \uy, \matX_{S_1} \} \le \text{Pr}(\mc{A}_d)$
\end{proposition}
\begin{proof}
  We will show that given the active set $S_1$, $d \in S_1$, $\uy$ and $\matX_{S_1}$, 
  the event $\{d \in S_H \}$, i.e., the decoded set of inactive variables contains $d$, implies
the event  $\mc{A}_d$. 
We first note that, since $|S_H| \le L+K-1$, there exists a set of $L_0=N-K-(L + K-1)$ inactive variables 
that do not belong to $S_H$.
Let $S_z \subset [N] \backslash S_1$ be such a set of inactive variables such that $|S_z| = L_0$ and
$S_z \cap S_H = \{ \emptyset \}$.

Further, since $d \in S_H$, this implies that there exits an $\omega \in \Omega_{\text{last}}$ 
such that $d$ belongs to $S_{\omega}$, where $\Omega_{\text{last}}$ is as defined in the decoding
scheme for $K > 1$ (see Section~\ref{sec_kgt1_case}).
With the notation described above, we can represent such $S_\omega$ as $\{d \cup S_{\omega d^c}\}$, where
$S_{\omega d^c} \subset [N] \backslash \{d\} \backslash S_z$ such that $|S_{\omega d^c}| = K-1$.
For any $\alpha \in S_z$, if we replace $d \in S_{\omega}$ with $\alpha$ and evaluate
$P(\uy|\matX_{\alpha}, \matX_{S_{\omega d^c}})$, it cannot be smaller than
$P(\uy|\matX_d, \matX_{S_{\omega d^c}})$ or else the decoding algorithm would have chosen $\alpha$
as belonging to $S_H$.
This implies that, there exists a realization of $\matX_{S_z}$ and $\matX_{S_{\omega d^c}}$ such that 
$P(\uy|\matX_{\alpha}, \matX_{S_{\omega d^c}}) \geq P(\uy|\matX_d, \matX_{S_{\omega d^c}}) ~ 
\forall ~ {\alpha} \in S_{z}$, i.e., ${\mc{A}_d}$ occurs. 
\end{proof}
We now upper bound $\text{Pr}(\mc{A}_d)$ as follows:
\begin{align} \label{eq:pA_bnd1}
  \text{Pr}(\mc{A}_d) \le \sum_{S_z \in \mathbf{\mc{S}_z}} ~  
  \sum_{S_{\omega d^c} \in \mathbf{\mc{S}_{\omega d^c}}} ~ q_d,
\end{align}
where $q_d \triangleq \text{Pr}\{ \mc{A}_d(S_z, S_{\omega d^c}) | S_1~ \mbox{is active set}, \uy, 
\matX_{S_1} \}$.
Here, the randomness comes from 
the set of variables in $S_z$ and $S_{\omega d^c}$, i.e., $\matX_{S_z}$ and $\matX_{S_{\omega d^c}}$.
  Let $s$ be such that $0 \le s \le 1$. We have
\restarthintedrel
\begin{align}
  \nonumber
  q_d &= \sum_{\matX_{S_z}, \matX_{S_{\omega d^c}} \in \mc{A}_d({S_z, S_{\omega d^c}})}
  Q(\matX_{S_z}, \matX_{S_{\omega d^c}}) \\
  \nonumber
  &\hintedrel[ina]{\le} \sum_{\matX_{S_{\omega d^c}},\matX_{S_z} \in \mc{A}_d({S_z, S_{\omega d^c}})}
  Q(\matX_{S_z}, \matX_{S_{\omega d^c}})
  \prod_{S_{\alpha} \in S_z} \left[ \frac{P(\uy|\matX_{\alpha}, \matX_{S_{\omega d^c}})}
    {P(\uy|\matX_{d}, \matX_{S_{\omega d^c}})} \right ]^s\\
    \nonumber
    &\hintedrel[inb]{\le} \sum_{\matX_{S_{\omega d^c}}}~ Q(\matX_{S_{\omega d^c}}) ~ 
    \sum_{ \matX_{S_z}} ~ Q(\matX_{S_z}) 
  \prod_{S_{\alpha} \in S_z} \left[ \frac{P(\uy|\matX_{\alpha}, \matX_{S_{\omega d^c}})}
    {P(\uy|\matX_{d}, \matX_{S_{\omega d^c}})} \right ]^s\\
    \nonumber
    &\hintedrel[inc]{=} \sum_{\matX_{S_{\omega d^c}}} Q(\matX_{S_{\omega d^c}})
    \prod_{l=1}^{L_0} \sum_{\matX_{S_{\alpha }}} Q(\matX_{S_{\alpha }}) 
    \left[ \frac{P(\uy|\matX_{\alpha}, \matX_{S_{\omega d^c}})}
    {P(\uy|\matX_{d}, \matX_{S_{\omega d^c}})} \right ]^s\\
    \nonumber
    &\hintedrel[ind]{=} \sum_{\matX_{S_{\omega d^c}}} Q(\matX_{S_{\omega d^c}})
    \left \{ \sum_{\matX_{S_{\alpha }}} Q(\matX_{S_{\alpha }}) 
    \left[ \frac{P(\uy|\matX_{\alpha}, \matX_{S_{\omega d^c}})}
      {P(\uy|\matX_{d}, \matX_{S_{\omega d^c}})} \right ]^s \right \} ^{L_0}\\
    \label{eq:q_bnd1}
    &= \sum_{\matX_{S_{\omega d^c}}} Q(\matX_{S_{\omega d^c}})
    \underbrace{
    \left \{ \sum_{\matX_{S_{\alpha }}} Q(\matX_{S_{\alpha }}) 
    \left[ \frac{P(\uy, \matX_{S_{\omega d^c}} |\matX_{\alpha})}
      {P(\uy, \matX_{S_{\omega d^c}}|\matX_{d} )} \right ]^s \right \} ^{L_0}
    }_{\triangleq \mc{P}_0(\uy, \matX_d, \matX_{S_{\omega d^c}})}.
\end{align}
In the above, (\hintref{ina})-(\hintref{ind}) follow using the same reasoning as in
(\ref{eq:k1uppbndq}) in the proof of Theorem \ref{Suff_K1Thm} (Section \ref{sec_proof_suffkeq1}).
We note that, due to symmetry in the construction of codebook, $\mc{P}_0(\uy, \matX_d, \matX_{S_{\omega d^c}})$
does not depend upon the index set $S_z$ or $\matX_{S_z}$. In fact, it depends only upon the
given realizations of $\matX_{S_{\omega d^c}}$, $\matX_d$ and not on the particular index sets
$S_{\omega d^c}$ and $d$, respectively.
Thus, from (\ref{eq:pA_bnd1}), and for some $0 \le \rho \le 1$, we get
\begin{align} \label{eq:pA_bnd2}
  \text{Pr}(\mc{A}_d) &\le \sum_{S_{\omega d^c} \in \mathbf{\mc{S}_{\omega d^c}}}
  \sum_{\matX_{S_{\omega d^c}}} Q(\matX_{S_{\omega d^c}})
  \left [ \sum_{S_z \in \mc{S}_z} \mc{P}_0(\uy, \matX_d, \matX_{S_{\omega d^c}}) \right ]\\
  &\le \sum_{S_{\omega d^c} \in \mathbf{\mc{S}_{\omega d^c}}}
  \sum_{\matX_{S_{\omega d^c}}} Q(\matX_{S_{\omega d^c}})
  \left [ \sum_{S_z \in \mc{S}_z} \mc{P}_0(\uy, \matX_d, \matX_{S_{\omega d^c}}) \right ]^{\rho}\\
  &\le {N - 1 - L_0 \choose K -1}
  \sum_{\matX_{S_{\omega d^c}}} Q(\matX_{S_{\omega d^c}})
  \left [ {N-K \choose L_0} \mc{P}_0(\uy, \matX_d, \matX_{S_{\omega d^c}}) \right ]^{\rho}.
\end{align}
The second inequality above follows since the expression inside the square 
brackets represents the probability of a union of events and thus, as in $K=1$ case, by raising it 
to a power $0 < \rho \le 1$, we still get an upper bound~\cite[Section~5.6]{Gallager_bk}.
Let $C_2 \triangleq {N-K \choose L_0}^\rho {N - 1 - L_0 \choose K -1}$.
Using proposition \ref{prop_err_up_bnd}, we substitute the above expression into (\ref{eq:overallpe_eq})
to get:
\restarthintedrel
 \begin{align}
        \nonumber
  \text{Pr}({\mc{E}_d}) &\le~
        C_2 \sum_{\uy} \sum_{\matX_{S_{1}}} Q(\matX_{S_1}) P(\uy | \matX_{S_1})  
        \sum_{\matX_{S_{\omega d^c}}} Q(\matX_{S_{\omega d^c}})
        \left [ \mc{P}_0(\uy, \matX_d, \matX_{S_{\omega d^c}}) \right ]^{\rho} \\
	\nonumber
	&\hintedrel[ina]{\le}~ 
        C_2 \sum_{\uy} \sum_{\matX_{S_{1}}} \sum_{\matX_{S_{\omega d^c}}} 
	Q(\matX_{S_1}) P(\uy, \matX_{S_{\omega d^c}} | \matX_{S_1})  
        \left [ \mc{P}_0(\uy, \matX_d, \matX_{S_{\omega d^c}}) \right ]^{\rho} \\
	\nonumber
	&\hintedrel[inb]{\le}~ 
	C_2 \sum_{\uy} \sum_{\matX_d} \sum_{\matX_{1 d^c}} \sum_{\matX_{S_{\omega d^c}}} 
	Q(\matX_d) P(\uy, \matX_{S_{\omega d^c}}, \matX_{S_{1 d^c}} | \matX_d)  
        \left [ \mc{P}_0(\uy, \matX_d, \matX_{S_{\omega d^c}}) \right ]^{\rho} \\
	\nonumber
	&\hintedrel[inc]{\le}~ 
        C_2 \sum_{\uy} \sum_{\matX_{S_{\omega d^c}}} \sum_{\matX_d} 
	  Q(\matX_d)~ P(\uy, \matX_{S_{\omega d^c}} | \matX_d) 
	   \left \{ \sum_{\matX_{S_{\alpha }}} Q(\matX_{S_{\alpha }}) 
	       \left[ \frac{P(\uy, \matX_{S_{\omega d^c}} |\matX_{\alpha})}
		       {P(\uy, \matX_{S_{\omega d^c}}|\matX_{d} )} \right ]^s \right \} ^{\rho L_0} \\
        \label{eq:final_pe}
	&\hintedrel[ind]{\le}~ 
        C_2 \sum_{\uy} \sum_{\matX_{S_{\omega d^c}}}
	   \left \{ \sum_{\matX_{S_{\alpha }}} Q(\matX_{S_{\alpha }}) 
	       P(\uy, \matX_{S_{\omega d^c}} |\matX_{\alpha})^{\frac{1}{1 + \rho L_0}}
	     \right \}^{ 1 + \rho L_0}.
\end{align}
In the above equation, (\hintref{ina}) follows by using the fact that given the active 
set $S_1$, $\uy$ is independent of the other input variables. 
Thus, $P(\uy, \matX_{S_{\omega d^c}}  | \matX_{S_1}) = P(\uy | \matX_{S_1}) 
Q(\matX_{S_{\omega d^c}})$. (\hintref{inb}) follows since
$S_1 = \{d \cup S_{1 d^c} \}$. (\hintref{inc}) follows by substituting the 
expression for $\mc{P}_0$ and by averaging out $\matX_{S_{1 d^c}}$, since the 
expression for $\mc{P}_0$ does not depend upon $\matX_{S_{1d^c}}$.
In (\hintref{inc}), the term $[P(\uy, \matX_{S_{\omega d^c}}|\matX_{d} )]^{s\rho L_0}$ can be factored
out from expression inside the curly braces.
Finally, (\hintref{ind}) is obtained by choosing $s = \frac{1}{1 + \rho L_0}$ and simplifying further.
Next, the above upper bound for $\text{Pr}(\mc{E}_d)$ depends only on $\matX_d$ and not
on any particular value of $d$. Thus, from (\ref{eq:overallpe_eq}) and (\ref{eq:final_pe}) we get:
\begin{align}
        \nonumber
  P_e &\le K C_2 \sum_{\uy} \sum_{\matX_{S_{\omega d^c}}}
	   \left \{ \sum_{\matX_{S_{\alpha }}} Q(\matX_{S_{\alpha }}) 
	       P(\uy, \matX_{S_{\omega d^c}} |\matX_{\alpha})^{\frac{1}{1 + \rho L_0}}
	     \right \}^{ 1 + \rho L_0} \\
  & \le \exp \left [ - M \left (
    E_0(\rho,1, L_0) - \frac{ \log (K C_2) }{M} \right ) \right ].
\end{align}
The inequality above is obtained by further simplifying
using independence across different observations and writing the bound in the exponential 
form, as in the $K=1$ case. 
The upper bound on $P_e$ given in~(\ref{eq:nonasym_pe}) now follows by substituting the value
of $C_2$ in the above. Hence the proof.~$\blacksquare$

\subsection{Proof of Theorem \ref{thm_necc_cond}: Necessary Number of Observations}
\label{sec_proof_necc}
For the purpose of this proof, recall that $P_e$ was defined 
in (\ref{eq:lb_pe_def}).
We need to prove that $\lim_{N \rightarrow \infty} P_e = 0$ implies the bound 
on the number of observations as given by (\ref{eq:necc_cond_M}). Towards that
end, we first find, by lower bounding $P_e$, the conditions
on $M$ that will lead to the error probability being bounded away from zero.
We consider a genie-aided lower bound, where we assume that the active set is
partially known. Let us define a partition for $S_{\omega}$ as 
$S_\omega = S^{(j)} \cup S^{(K-j)}$, where $|S^{(j)}| = j$ and
$|S^{(K-j)}| = K - j$ and $S^{(j)} \cap S^{(K-j)} = \{ \emptyset\}$. We assume that
$S^{(K-j)}$ (and hence, for a given code, the matrix $\matX_{S^{(K-j)}}$) is known to us. For the result to follow, by symmetry of the codebook construction, it does not matter which of the $K-j$ indices in the defective set are assumed to be known. 
Now consider $H(\omega, E | \uy, \matX_{S^{(K-j)}})$:
\restarthintedrel
\begin{align} \label{eq:fano_chrule_1}
  H(\omega, E | \uy, \matX_{S^{(K-j)}}) &= H( E | \uy, \matX_{S^{(K-j)}}) + 
  H( \omega | E, \uy, \matX_{S^{(K-j)}}) \\
  &\hintedrel[feq1]{\le} H_b(P_e) + (1 - P_e) H( \omega | E=0, \uy, \matX_{S^{(K-j)}}) + 
  P_e H( \omega | E = 1, \uy, \matX_{S^{(K-j)}}) \\
  &\hintedrel[feq2]{\le} H_b(P_e) + (1 - P_e) \log {N- K + j - L \choose j}  + 
  P_e H(\omega | \matX_{S^{(K-j)}}) \\
  \label{eq:fano_chrule_11}
  &\hintedrel[feq3]{\le} H_b(P_e) + (1 - P_e) \log {N- K + j - L \choose j}  + 
  P_e \log {N - K + j \choose j}.
\end{align}

In the above, (\hintref{feq1}) follows since $E$ is a binary RV and 
$H( E | \uy, \matX_{S^{(K-j)}}) \le H(E) = H_b(P_e) \le 1$.
Since the entropy of any RV is bounded by the logarithm of the alphabet size,
(\hintref{feq2}) follows by considering the cardinality of the remaining number of outcomes
conditioned on the outcome of $E$. For example, when $E=0$, i.e., 
when there is no error, the
number of ways of choosing the set $S^{(j)}$ is  
${N- K + j - L \choose j}$.
(\hintref{feq3}) follows by using a trivial bound on $H(\omega | \matX_{S^{(K-j)}})$.
Also,  
\begin{align} \label{eq:fano_chrule_2}
  H(\omega, E | \uy, \matX_{S^{(K-j)}}) &= H(\omega | \uy, \matX_{S^{(K-j)}}) + 
  H( E | \omega , \uy, \matX_{S^{(K-j)}})
  = H(\omega | \uy, \matX_{S^{(K-j)}}).
\end{align}
For a given $\matX$, the mapping from $\omega$ to $\matX_{S_\omega}$ is one-one
and onto. Thus, $H(\omega | \matX_{S^{(K-j)}}) = 
H(\matX_{S_{\omega}} | \matX_{S^{(K-j)}})$ and similarly 
$H(\omega | \uy, \matX_{S^{(K-j)}}) = 
H(\matX_{S_{\omega}} | \uy, \matX_{S^{(K-j)}})$.
Using the above and the fact that 
$H(\omega | \matX_{S^{(K-j)}}) = \log {N - K + j \choose j}$
in (\ref{eq:fano_chrule_11}) and (\ref{eq:fano_chrule_2}), we get
\begin{align} \label{eq:pe_lower_bnd}
  \log {N - K + j \choose j} &= H(\matX_{S_{\omega}} | \uy, \matX_{S^{(K-j)}}) +
  I(\matX_{S_{\omega}} ; \uy| \matX_{S^{(K-j)}}) \\
  &\le H_b(P_e) + \log {N- K + j - L \choose j} + P_e \Gamma_l(L, N, K, j)
  + I(\matX_{S_{\omega}} ; \uy| \matX_{S^{(K-j)}}).
\end{align}
Note that $ I(\matX_{S_{\omega}} ; \uy| \matX_{S^{(K-j)}}) = 
 I(\matX_{S^{(j)}} ; \uy| \matX_{S^{(K-j)}})$ and 
using basic properties of entropy, mutual information and the i.i.d.\ assumption across
observations, it can be shown that \cite{Atia_BooleanCS}:
\begin{align} \label{eq:I_vectoscalar}
  I(\matX_{S^{(j)}} ; \uy| \matX_{S^{(K-j)}}) \le M 
  I(X_{S^{(j)}} ; Y| X_{S^{(K-j)}}) =  MI^{(j)}.
\end{align}

Thus, we get a genie aided lower bound on the probability of error as
\begin{align} \label{eq:pe_lower_bnd_final}
  P_e &\ge 1 - \frac{ H_b(P_e) + M I^{(j)}}
  {\Gamma_l(L, N, K, j)} ~ ~ ~ \forall~ j = 1, 2, \ldots, K.
\end{align}
This further implies
\begin{align} \label{eq:pe_lower_bnd_M}
  M &\ge \frac{(1 - P_e) \Gamma_l(L, N, K, j) - H_b(P_e) }{I^{(j)}}
  ~ ~ ~ \forall~ j = 1, 2, \ldots, K.
\end{align}
The above equation holds for all $j = 1, 2, \ldots, K$ and thus, the lower bound on the 
number of observations follow easily by noting that $H_b(P_e) \rightarrow 0$ as 
$P_e \rightarrow 0$. Hence the proof.
\subsection{Proof of Theorem \ref{lemma_suff_NN_KON}} \label{proof_suff_NN_KON}
  In (\ref{eq:nonasym_pe}), consider the term  
  $\mc{T}(\rho) \triangleq \left ( M E_0(\rho,1, L_0) - { \rho \log {N-K \choose L_0} } - 
    { \log \left[K {N - 1 - L_0 \choose K -1} \right]} \right )$. 
Using the results of Lemma~\ref{E0LB_Thm}, 
    for any $\epsilon_0 > 0$, at $\rho = \rho_0$ where $\rho_0 = \frac{K-1}{L_0}$,\footnote{Note that, 
    for $L \le N-3K+1$, $\rho_0 = \frac{K-1}{L_0} < 1$.}
    if $M$ is chosen as
    \begin{align} \label{eq:Suff_cond_0}
          M >  (1+\epsilon_0) \left [
	    \frac{ \rho_0 \log {N-K \choose L_0} }{E_0^{(lb)}} +
	    \frac{ \log \left[ {L + 2(K-1) \choose K -1} \right]}{E_0^{(lb)}} + 
	    \frac{ \log K }{E_0^{(lb)}} \right],
    \end{align}
    then, 
    $\mc{T}(\rho) > \epsilon_0 (K-1) \lb \log \frac{N-K}{L_0}  + \log ( 2+ \frac{L}{K-1})\rb > 
    \epsilon_0 (K-1) \log \frac{N-K}{L_0} > 0$.

  Using Stirling's formula, for any $n \in \mathbb{Z}_{+}$: $\sqrt{2 \pi} n^{n + 1/2} e^{-n} \le n! \le 
  e n^{n + 1/2} e^{-n}$, we note
  \begin{align}
    \log {N-K \choose L_0} &\le L_0 \log (\frac{N-K}{L_0}) + (L+K-1) \log (\frac{N-K}{L+K-1})
    + \frac{1}{2} \log \frac{N-K}{L_0 (L+K-1)}\\
    &\le L_0 \log (\frac{N-K}{L_0}) + (L+K-1) \log (\frac{N-K}{L+K-1}).
  \end{align}
  The second inequality follows since under the assumptions on the range of $L$,
  $\frac{N-K}{L_0 (L+K-1)} < 1$. Thus, with $\theta_0 \triangleq \frac{L+K-1}{N-K}$, we get
  $\frac{\log {N-K \choose L_0}}{L_0} \le \frac{H_b(\theta_0)}{1 - \theta_0}$.
  Finally, the bound in (\ref{eq:suff_NN_KON}) results by using the inequality ${m \choose n} \le \left (\frac{e m}{n} \right )^n$ 
  to upper bound the second term in (\ref{eq:Suff_cond_0}).

\section{Conclusions} \label{sec_conclusions}
In this paper, we considered the problem of identifying $L$ non-defective items out of 
a large population of $N$ items containing $K$ defective items in a general 
sparse signal modeling setup. We contrasted two approaches: identifying the 
defective items using the observations followed by picking $L$ items from the 
complement set, and  directly identifying non-defective items from the observations. 
We derived  upper and lower bounds on the number of observations required for identifying 
the $L$ non-defective items. We showed that a gain in the number of 
observations is obtainable by directly identifying the non-defective items. 
We also applied the results in a nonadaptive group 
testing setup. We characterized the number of tests that are sufficient 
to identify a subset of non-defective items in a large population, under both dilution 
and additive noise models. 
Our results were information theoretic in nature, without considering the 
practicability of the decoding algorithms. Our companion study looks at 
finding computationally tractable algorithms for directly identifying a subset 
of inactive variables, in the context of non-adaptive group testing. Future work could focus on tightening the upper bounds on the sufficient number of tests, thereby obtaining order-optimal results.

\appendix
\subsection{Proof of Lemma \ref{E0LB_Thm} } \label{sec_proof_E0}
   From (\ref{eq:E0_def}), it follows that:
   \begin{align} \label{eq:E0_def_1} 
    E_0(\rho,j, n) &= - \log \sum_{Y \in \mathcal{Y}}  \sum_{X_{S^{(K-j)}} \in \mathcal{X}^{K-j}} 
	  Q\left( X_{S^{(K-j)}} \right) \left \{ \sum_{X_{S^{(j)}} \in \mathcal{X}^j} Q(X_{S^{(j)}}) 
          \left( P(Y | X_{S^{(K-j)}}, X_{S^{(j)}})\right)^{\frac{1}{1 + \rho n}}
	       \right \}^{ 1 + \rho n} 
   \end{align}
   In the above, we substitute $j=1$, $n = L_0$ and $\rho = \rho_0$.
   Let $w(X_{S^{(K-1)}})$ denote the number of $1$'s in 
    $X_{S^{(K-1)}} \in \{0,1\}^{(K-1)}$. Let $n_0 \triangleq 1 + \rho_0 L_0$ and further, note
    that $n_0 = K$.
    For the non-adaptive group testing signal model, using (\ref{eq:gtmodel}), we have computed
    the posterior probability $P(Y | X_{S^{(K-1)}},X_{S^{(1)}})$ for different scenarios and summarized it in
    Table \ref{tab:Pytab}. 
  \begin{table*}[t]
    \centering
    \caption{$P(Y | X_{S^{(K-1)}},X_{S^{(1)}})$ values under different scenarios for the non-adaptive group testing signal model.}
    \begin{tabular}{|c|l|l|l|l|} \hline 
      & \multicolumn{2}{|c|}{$w(X_{S^{(K-1)}})=0$} & \multicolumn{2}{|c|}{$w(X_{S^{(K-1)}})=l, 1\le l \le K-1$} \\ \cline{2-5}
        & $X_{S^{(1)}} = 0$ & $X_{S^{(1)}} = 1$ & $X_{S^{(1)}} = 0$ & $X_{S^{(1)}} = 1$ \\ \hline
	$P(Y=0 | X_{S^{(K-1)}},X_{S^{(1)}})$ & $(1-q)$ & $(1-q)u$ & $(1-q)u^l$ & $(1-q)u^{l+1}$ \\ \hline
	$P(Y=1 | X_{S^{(K-1)}},X_{S^{(1)}})$ & $q$ & $(1-(1-q)u)$ & $1 - (1-q)u^l$  & $1 - (1-q)u^{l+1}$  \\ \hline
				\end{tabular}
				\label{tab:Pytab}
			      \end{table*}
\begin{enumerate}[(a)]
  \item Noiseless case: Using $q=0$, $u=0$ in Table~\ref{tab:Pytab} and substituting in (\ref{eq:E0_def_1}) we get:
\begin{align} \label{eq:E0_NN}
  E_0(\rho, 1, L_0) = -\log \left[ 1 - (1 - p)^{(K-1)}\left(1 - (1 - p)^{n_0} - p^{n_0} \right )\right ].
\end{align}
Using, (i) the inequality $-\log(1-x) \ge x$ for $x < 1$, (ii) For $p=\frac{1}{K}$, $(1-p)^{(K-1)} > e^{-1}$ and 
$(1-p)^{K} < e^{-1}$, 
(\ref{eq:e0_lb_nonoise}) results.

  \item Additive noise case: Using $u=0$ in Table~\ref{tab:Pytab} and substituting in (\ref{eq:E0_def_1}) we get:
\begin{align} \label{eq:E0_AN}
  E_0(\rho, 1, L_0) = -\log \left[ 1 - (1 - p)^{(K-1)}\left( 1 - (1-q)(1-p)^{n_0} - 
    \left \{ (1-p) q^{\frac{1}{n_0}} + p \right \}^{n_0} \right )\right ].
\end{align}
To lower bound $E_0$, we first upper bound the term $t_0 \triangleq \left \{ (1-p) q^{\frac{1}{n_0}} + p \right \}^{n_0}$. 
For any $n \ge 1$, $x^n$ is a convex function and hence, using Jensen's inequality we get 
$t_0 \le (1-p) q + p$. Substituting and further simplifying we get:
\begin{align} \label{eq:E0_AN_1}
  E_0(\rho, 1, L_0) &\ge -\log \left[ 1 - (1 - p)^{K}(1-q)\left( 1 - (1-p)^{(n_0-1)} \right )\right ].
\end{align}
The bound in (\ref{eq:e0_lb_q}) now results by using the inequality $-\log(1-x) \ge x$ for $x < 1$ 
and noting the following:
For $p=\frac{1}{K}$, using the inequality, $1-x \le e^{-x} \le 1 - \frac{x}{2}$ for $0\le x \le 1$, 
we get $(1-p)^K \ge e^{-2}$ and $1 - (1-p)^{(n_0-1)} \ge \frac{n_0 - 1}{2 K} \ge \frac{1}{4}$ for $K \ge 2$.

  \item Dilution noise case: 
    Let $G_l \triangleq {K-1 \choose l} p^l (1-p)^{(K-1-l)}$. Using $q=0$ in Table~\ref{tab:Pytab} and 
    substituting in (\ref{eq:E0_def_1}) we get:
\begin{align} \label{eq:E0_DN}
  E_0(\rho, 1, L_0) &= -\log \left[ T_0 + T_1 \right ], ~ ~ \mbox{where},\\
  \nonumber
  T_0 \triangleq \sum_{l=0}^{K-1} G_l u^l \left( (1-p) + pu^{\frac{1}{n_0}}\right)^{n_0} ~ &\mbox{and} ~
  T_1 \triangleq \sum_{l=0}^{K-1} G_l 
  {\left( (1-p)(1 - u^l)^{\frac{1}{n_0}} + p(1 - u^{l+1})^{\frac{1}{n_0}} \right)^{n_0}}.
\end{align}
Using Jensen's inequality to upper bound $T_1$, we get
\begin{align} \label{eq:E0_DN_1}
  T_1 &\le \sum_{l=0}^{K-1} G_l \left( (1-p)(1 - u^l) + p(1 - u^{l+1}) \right)\\
  &= 1 - \zeta_0 \sum_{l=0}^{K-1} G_l u^l,	
\end{align}
where $\zeta_0 \triangleq (1 - (1-u) p)$ and we have made use of the fact that $\sum_{l=0}^{K-1} G_l = 1$.
Further, since $\sum_{l=0}^{K-1} G_l u^l = \zeta_0^{(K-1)}$, we get
\begin{align} \label{eq:E0_DN_2}
  E_0(\rho, 1, L_0) &\ge -\log \left[ 1 - (\zeta_0 - \psi_0) \zeta_0^{(K-1)} \right ],
\end{align}
where $\psi_0 \triangleq \left( 1- (1 - u^{\frac{1}{n_0}})p \right)^{n_0}$.
Using the inequality $-\log(1-x) \ge x$ for $x < 1$, we get:
\begin{align} \label{eq:E0_DN_3}
  E_0(\rho, 1, L_0) &\ge (\zeta_0 - \psi_0) \zeta_0^{(K-1)}
  \ge \left [ 1 - \left ( 1 - (1 - u^{\frac{1}{n_0}})p \right )^{n_0 - 1} \right ] \zeta_0^K,
\end{align}
where the second inequality follows since $( 1 - (1 - u^{\frac{1}{n_0}})p) \ge \zeta_0$.
The bound in (\ref{eq:e0_lb_u}) now results by noting the following:
For $p=\frac{1}{K}$, using the inequality, $1-x \le e^{-x} \le 1 - \frac{x}{2}$ for $0\le x \le 1$, 
we get $\zeta_0^K \ge e^{-2(1-u)} \ge e^{-2}$ and 
$\left [ 1 - \left ( 1 - (1 - u^{\frac{1}{n_0}})p \right )^{n_0 - 1} \right ]  
\ge (1-u^{\frac{1}{n_0}}) \frac{n_0 - 1}{2 K} \ge \frac{1}{4}(1-u^{\frac{1}{n_0}})$ for $K \ge 2$.

{\bf Remark:} For $\rho_0 = \frac{a}{L_0}$ for any $a$, $n_0 = 1+a$. Thus, 
$E_0(\rho_0, 1, L_0) \ge \frac{(1-u^{\frac{1}{1+a}}) a}{2 K}$. In particular, with $a = 1$, 
$E_0(\rho_0, 1, L_0) \ge \frac{(1-u^{\frac{1}{2}}) }{2 K}$.

\end{enumerate}
\subsection{Order-Tight Results for Necessary and Sufficient Number of Tests with Group Testing}
\label{order_res_proof}
In this section, we present a brief sketch of the derivation of the order results for
the necessary number of tests presented in Table \ref{tab:tab_necc_order}. We first
note that $I^{(j)} = H(Y | X_{S^{(K-j)}}) - H(Y | X_{S^{(K-j)}}, X_{S^{(j)}})$
\cite{Atia_BooleanCS}, where
$H(\cdot| \cdot )$ represents the entropy function\cite{Cover_bk}. 
From (\ref{eq:gtmodel}), we have
\begin{align}
H(Y | X_{S^{(K-j)}}) &= \sum_{l=0}^{K-j}\left [ {K-j \choose l} p^l (1 - p)^{K-j-l}
H_b\left ( (1-q) u^l (1 - p(1-u))^j \right) \right ] \\
H(Y | X_{S^{(K-j)}}, X_{S^{(j)}}) &= \sum_{i=0}^{K}\left [ {K \choose i} p^i (1 - p)^{K-i}
H_b\left ( (1-q) u^i \right) \right ].
\end{align}
We use the results from \cite{johnson_gtbp} for bounding the mutual information term. 
We collect the required results from \cite{johnson_gtbp} in the following lemma.
\begin{lemma}
  Bounds on $I^{(j)}$\cite{johnson_gtbp}: Let $p = \frac{\delta}{K}$. $I^{(j)}$ can be
  expressed as $I_1^{(j)} + I_2^{(j)}$, where 
  \begin{align} \label{eq:I1_1}
    I_1^{(j)} = \delta e^{-\delta(1-u)} (1 - q) \left(u\log u + 1-u\right) \frac{j}{K} + 
    O\left(\frac{1}{K^2}\right). 
  \end{align}
  For the case with $u=0$ and $q > 0$ we have:
  \begin{align} \label{eq:I2_1}
    I_2^{(j)} = \delta e^{-\delta} \left( \log(\frac{1}{q}) - (1 - q) \right) \frac{j}{K} + 
    O\left(\frac{1}{K^2}\right),
  \end{align}
  and for $q = 0$, $u \ge 0$ we have:
  \begin{align} \label{eq:I2_2}
    \nonumber
    \delta e^{-\delta} &\left( (1-u)\left[ \log\frac{K}{j \delta(1-u)} \right] - u \right) \frac{j}{K} + 
    O\left(\frac{1}{K^2}\right) ~ \le ~ I_2^{(j)} \\
    &\le~ \delta e^{-\delta(1 - u^2)} \left( (1-u)\left[ \log \frac{K}{j \delta(1-u)} \right]
    - u + u^2 \right) \frac{j}{K} + O\left(\frac{1}{K^2}\right).
  \end{align}
\end{lemma}
Thus, with $\delta = 1$ and large $K$, neglecting $O(1/K^2)$ terms, we get: (a) For $u=0$, $q > 0$ case,
$I^{(j)} \approx \frac{j}{eK}\log(\frac{1}{q})$. 
(b) For $q=0$, $ 0 \le u \le 0.5$ case, simplifying further, we get
\begin{align}
  \frac{j}{eK}(1-u) \log\frac{K}{j} \lessapprox I^{(j)} \lessapprox \frac{j}{e^{1/2}K}(1-u) 
  \left(\log\frac{K}{j} + 1 \right).
\end{align}
In the above, we have used the notation ``$\approx$'' and ``$\lessapprox$'' to highlight the fact that
$O(\frac{1}{K^2})$ terms have been neglected in the above expressions for $I^{(j)}$.
The order results for lower bounds now follow by first noting that 
$\max_{1 \le j \le K} \frac{\Gamma_l(L, N, K, j)}{I^{(j)}} \ge \frac{\Gamma_l(L, N, K, 1)}{I^{(1)}}$, 
and, for the scaling regimes under consideration the combinatorial term, $\Gamma_l(L,N,K,1) $ can be asymptotically 
bounded as $\lim_{N \rightarrow \infty} \Gamma_l(L, N, K, 1) \ge \log\frac{1}{1-\alpha_0}$.

\bibliographystyle{IEEEbib}
\bibliography{IEEEabrv,bibJournalList,refs_fnds}
%
%
%
%
%
%
\end{document}